\documentclass[onefignum,onetabnum]{siamart171218}

\usepackage{lipsum}
\usepackage{amsmath,amssymb,amsfonts}
\usepackage{bm}
\usepackage{graphicx}
\usepackage{epstopdf}
\usepackage{algorithmic}
\usepackage{subfig}
\usepackage{float}
\usepackage{mathtools}
\ifpdf
  \DeclareGraphicsExtensions{.eps,.pdf,.png,.jpg}
\else
  \DeclareGraphicsExtensions{.eps}
\fi

\newcommand*\diff{\mathop{}\!\mathrm{d}} 

\DeclarePairedDelimiter\floor{\lfloor}{\rfloor}

\newsiamremark{remark}{Remark}
\newsiamremark{hypothesis}{Hypothesis}
\newsiamremark{example}{Example}
\crefname{hypothesis}{Hypothesis}{Hypotheses}
\newsiamthm{claim}{Claim}

\headers{Non-adiabatic transitions in multiple dimensions}{V. Betz, B. D. Goddard and T. Hurst}

\title{Non-adiabatic transitions in multiple dimensions\thanks{Submitted to the editors DATE.
\funding{T. Hurst was supported by The Maxwell Institute Graduate School in Analysis and its Applications, a Centre for Doctoral Training funded by the UK Engineering and Physical Sciences Research Council (grant EP/L016508/01), the Scottish Funding Council, Heriot-Watt University and the University of Edinburgh.}}}

\author{V. Betz \thanks{Fachbereich Mathematik, Technische Universit\"at Darmstadt, 64289 Darmstadt, Germany (\email{betz@mathematik.tu-darmstadt.de}, \url{http://www.mathematik.tu-darmstadt.de/\~betz/}).}
\and B. D. Goddard\thanks{School of Mathematics and the Maxwell Institute for Mathematical Sciences, University of Edinburgh, Edinburgh, UK, EH9 3FD (\email{b.goddard@ed.ac.uk}, \url{http://www.maths.ed.ac.uk/\~ bgoddard/}).}
\and T. Hurst\thanks{Maxwell Institute for Mathematical Sciences, School of Mathematics, University of Edinburgh, Edinburgh, UK, EH9 3FD
  (\email{t.hurst@sms.ed.ac.uk}, \url{http://www.maxwell.ac.uk/migsaa/people/tim-hurst}).}
}
\usepackage{amsopn}

\ifpdf
\hypersetup{
	pdftitle={An Example Article},
	pdfauthor={D. Doe, P. T. Frank, and J. E. Smith}
}
\fi

\begin{document}
	\maketitle
%
	\begin{abstract}
		We consider non-adiabatic transitions in multiple dimensions, which occur when the Born-Oppenheimer approximation breaks 
		down. We present a general, multi-dimensional algorithm which can be used to accurately and efficiently compute the transmitted 
		wavepacket at an avoided
		crossing.  The algorithm requires only one-level Born-Oppenheimer dynamics and local knowledge of the potential surfaces.
		Crucially, in contrast to standard methods in the literature, we compute the whole wavepacket, including its phase,
		rather than simply the transition probability.  We demonstrate the excellent agreement with full quantum dynamics for a 
		range of examples in two dimensions.  We also demonstrate surprisingly good agreement for a system with a full conical 
		intersection.
	\end{abstract}
%
	\begin{keywords}
		time-dependent Schr\"{o}dinger equation, non-adiabatic transitions, superadiabatic representations.
	\end{keywords}
%
	\begin{AMS}
		35Q40, 81V55
	\end{AMS}
	\section{Introduction}\label{Section:Introduction}
	Many computations in quantum molecular dynamics rely on the Born-Oppenheimer Approximation (BOA) \cite{BornOppenheimer27}, which utilises the small ratio $\varepsilon^2$ of electronic and reduced nuclear masses to replace the electronic degrees of freedom with {\em Born-Oppenheimer} potential surfaces. When these surfaces are well separated, the BOA further reduces computational complexity by decoupling the dynamics to individual surfaces.
	
	However, there are many physical examples \cite{Domcke04},\cite{Domcke11},\cite{Nakamura12} and \cite{Tully12} where the Born-Oppenheimer surfaces are not well separated or even have a full intersection. In these regions the BOA breaks down, and the coupled dynamics must be considered; when a wavepacket travels over a region where the surfaces are separated by a small but none-vanishing amount, a chemically crucial portion of the wavepacket can move to a different energy level via a {\em non-adiabatic transition}. The existence of the small parameter $\varepsilon$ introduced several challenges when attempting to numerically approximate the dynamics.
First, and independent of the existence of an avoided or full crossing, the wavepacket oscillates with frequency $1/\varepsilon$ and hence
a very fine computational grid is required.  Furthermore, in the region of an avoided crossing, the dynamics produce rapid oscillations;  the transmitted wavepacket very close to the crossing is $\mathcal{O}(\varepsilon)$, but in the scattering regime the transmission is exponentially small.  It is therefore necessary to travel far from the avoided crossing with a small time-step to accurately calculate the phase, size and shape of the transmitted wavepacket.  In order to calculate the exponentially small wavepacket, one must ensure that the absolute errors in a given numerical scheme are also exponentially small, or they will swamp the true result.  Finally, the number of gridpoints in the domain increases exponentially as the dimension of the system increases. Thus standard numerical algorithms quickly become computationally intractable.
	
	Many efforts have been made to avoid computational expense by approximating the transmitted wavepacket while avoiding the coupled dynamics. {\em Surface hopping algorithms} discussed in \cite{Tully71,MillerGeorge72,Stine76,Kuntz79,Tully90,HammesSchifferTully94,MullerStock97,FabianoGroenhofThiel08,
		FermanianKammererLasser08,LasserSwart08,BelyaevLasserTriglia14,BelyaevDomckeLasserTriglia15}
	 approximate the transition using classical dynamics, where the {\em Landau-Zener transition rate} \cite{Zener32}, \cite{Landau65} is used to determine the size of the transmitted wavepacket. This method has enjoyed some success, and has recently been applied to higher dimensional systems \cite{LasserSwart08,BelyaevLasserTriglia14}. However, the full transmitted quantum wavepacket is not calculated; phase information is lost. Such information is crucial when considering systems with interference effects, e.g.\ ones in which the initial wavepacket makes multiple transitions through an avoided crossing.  Recently, there have been efforts to include phase information in surface hopping algorithms \cite{ChaiJinLiMorandi15}. In contrast, in \cite{BetzGoddardTeufeul09} and \cite{BetzGoddard09}, a formula is derived to accurately approximate the full transmitted wavepacket, in one dimension, using only decoupled dynamics. The formula has been applied to a variety of examples with accurate results, including the transmitted wavepacket due to photo-dissociation of sodium iodide \cite{BetzGoddardManthe16}.
	
	In this paper we construct a method to apply the formula derived in \cite{BetzGoddardTeufeul09} and \cite{BetzGoddard09} to higher dimensional problems. We begin in  \cref{Section:Superadiabatic,Section:1DFormula} by outlining the derivation of the formula in one dimension \cite{BetzGoddardTeufeul09}, which involves deriving and approximating algebraic differential recursive equations for the quantum symbol of the coupling operator in {\em superadiabatic representations}. We extend these derivations to $d$ dimensions in \cref{Section:NDCouplingOperators}. In \cref{Section:SlicingAlgorithm} we create an $d$-dimensional formula for systems which are slowly varying in all but one dimension, then extend this result via a simple algorithm to obtain a general $d$-dimensional formula. We provide some examples and results in \cref{Section:Numerics} and note conclusions and future work in \cref{Section:Conclusions}.
	
	\section{Superadiabatic Representations}\label{Section:Superadiabatic}		
	Consider the evolution of a semiclassical wavepacket in $d$ dimensions with $\bm{x}\in\mathbb{R}^d$ at time $t$, $\psi=\left(\begin{smallmatrix} \psi_1(\bm{x},t) \\ \psi_2(\bm{x},t)\end{smallmatrix}\right)$, governed by the following equation:
	\begin{align}
	i\varepsilon\partial_t\psi(\bm{x},t)=H\psi(\bm{x},t),\label{eq:DiabaticSchrodinger}
	\end{align}
	where $\varepsilon^2$ is the ratio between an electron and the reduced nuclear mass of the molecule, {\em i.e.} $\varepsilon<<1$. This system is derived after a rescaling of a two level Schr\"{o}dinger equation \cite{GoddardHurst18}. Note that as we are considering semiclassical wavepackets, the derivatives of which are of order $1/\varepsilon$. The Hamiltonian of a two level system is given by \cite{BetzGoddard09}
	\begin{align}
	H&=-\frac{\varepsilon^2}{2}\nabla_{\bm{x}}^2I+V(\bm{x})+d(\bm{x})I,\label{eq:DiabaticHamiltonian}
	\end{align}
	where
	\begin{align}
		V(\bm{x})=\begin{pmatrix}
		Z(\bm{x}) & X(\bm{x}) \\
		X(\bm{x}) & -Z(\bm{x})
		\end{pmatrix}\label{eq:DiabaticPotential}
	\end{align}
	and $d(\bm{x})$ is the part of the potential operator with non-zero trace. In general $V(\bm{x})$ can be given by a Hermitian matrix, but as noted in \cite{Berry90}, any Hermitian $V(\bm{x})$ can be transformed into real symmetric form. It is useful to define $\theta(\bm{x})=\tan^{-1}\left(\frac{X(\bm{x})}{Z(\bm{x})}\right)$, so that
	\[
		\cos\left(\theta(\bm{x})\right)=\frac{Z(\bm{x})}{\sqrt{X(\bm{x})^2+Z(\bm{x})^2}}, \qquad 
	 \sin\left(\theta(\bm{x})\right)=\frac{X(\bm{x})}{\sqrt{X(\bm{x})^2+Z(\bm{x})^2}}.
	 \]
	Then, defining $\rho(\bm{x})=\sqrt{X(\bm{x})^2+Z(\bm{x})^2}$, gives
\begin{align}
	V(\bm{x})&=\rho(\bm{x})\begin{pmatrix} \cos(\theta(\bm{x})) & \sin(\theta(\bm{x})) \\
	\sin(\theta(\bm{x})) & -\cos(\theta(\bm{x})) 
	\end{pmatrix}.\label{eq:Potential}
	\end{align}
	This is known as the {\em diabatic representation} of the system. We define $V_1=Z(\bm{x})+d(\bm{x})$ and $V_2=-Z(\bm{x})+d(\bm{x})$ as the two {\em diabatic potentials}, with the {\em diabatic coupling element} as the off-diagonal element $V_{12}=X(\bm{x})$. Consider the unitary matrix $U_0$ which diagonalises the potential operator $V(x)$:
	\begin{align}
	U_0(\bm{x})=\begin{pmatrix}
	\cos\left(\frac{\theta(\bm{x})}{2}\right) & \sin\left(\frac{\theta(\bm{x})}{2}\right)\\
	\sin\left(\frac{\theta(\bm{x})}{2}\right) & -\cos\left(\frac{\theta(\bm{x})}{2}\right)
	\end{pmatrix}.\label{eq:AdiabaticUnitary}
	\end{align}
	If we define $\phi_0(\bm{x})=U_0(\bm{x})\phi(\bm{x})$, then we arrive at the {\em adiabatic Schr\"{o}dinger equation}
	\begin{align}
	i\varepsilon\partial_t\psi_0(\bm{x},t)=H_0\psi_0(\bm{x},t).\label{eq:AdiabaticSchrodinger}
	\end{align}
	Here $H_0=U_0HU_0^{-1}$ is given by
	\begin{align}
	H_0&=-\frac{\varepsilon^2}{2}\nabla_{\bm{x}}^2+\begin{pmatrix}
	\rho(\bm{x})+d(\bm{x})+\varepsilon^2\frac{\|\nabla_{\bm{x}}\theta(\bm{x})\|^2}{8} & -\varepsilon\frac{\nabla_{\bm{x}}\theta(\bm{x})}{2}\cdot(\varepsilon\nabla_{\bm{x}})-\varepsilon^2\frac{\nabla_{\bm{x}}^2\theta(\bm{x})}{4}\\
	\varepsilon\frac{\nabla_{\bm{x}}\theta(\bm{x})}{2}\cdot(\varepsilon\nabla_{\bm{x}})+\varepsilon^2\frac{\nabla_{\bm{x}}^2\theta(\bm{x})}{4} & -\rho(\bm{x})+d(\bm{x})+\varepsilon^2\frac{\|\nabla_{\bm{x}}\theta(\bm{x})\|^2}{8}
	\end{pmatrix}.\label{eq:AdiabaticHamiltonian}
	\end{align}
	The {\em adiabatic potential surfaces} are given by the diagonal entries of the adiabatic potential matrix to leading order,
	\begin{align}
	V_U=\rho(\bm{x})+d(\bm{x}),\quad V_L=-\rho(\bm{x})+d(\bm{x}),\label{eq:AdiabaticSurfaces}
	\end{align}
	where $V_U$ is the upper adiabatic potential surface, and $V_L$ is the lower adiabatic potential surface. Assuming the initial wavepacket $\phi$ is purely on the upper level, the adiabatic representation approximates the transmitted wavepacket to leading order by the perturbative solution \cite{Teufel03}
	\begin{align}
	\psi_0^-(t)=-i\varepsilon\int_{-\infty}^te^{-\frac{i}{\varepsilon}(t-s)H^-}\kappa_1^-(\bm{x})\cdot(\varepsilon\partial_{\bm{x}})e^{-\frac{i}{\varepsilon}sH^+}\phi \diff s,\label{eq:AdiabaticPert}
	\end{align} 
	where
	\begin{gather}
	H^\pm=-\frac{\varepsilon^2}{2}\nabla_x^2\pm\rho(\bm{x})+d(\bm{x}),\label{eq:AdiaTerms}\quad	\kappa_1^\pm(\bm{x})=\pm\frac{\partial_{\bm{x}}\theta(\bm{x})}{2}.
	\end{gather}
	The perturbative solution in the adiabatic representation does not offer much explanation to the properties of the transmitted wavepacket. For instance, the constructed wavepacket at first looks to be $\mathcal{O}(\varepsilon)$. However due to the adiabatic coupling operator $\kappa_1^\pm$, fast oscillations and cancellations between upper and lower transmissions occur near the avoided crossing, so that far from  the crossing the transmitted wavepacket is much smaller than the transition at the crossing point (\cref{fig:MassDown}).
	\begin{figure}[ht!]
		\centering
		\hspace{-18mm}
		\includegraphics[width=0.5\textwidth]{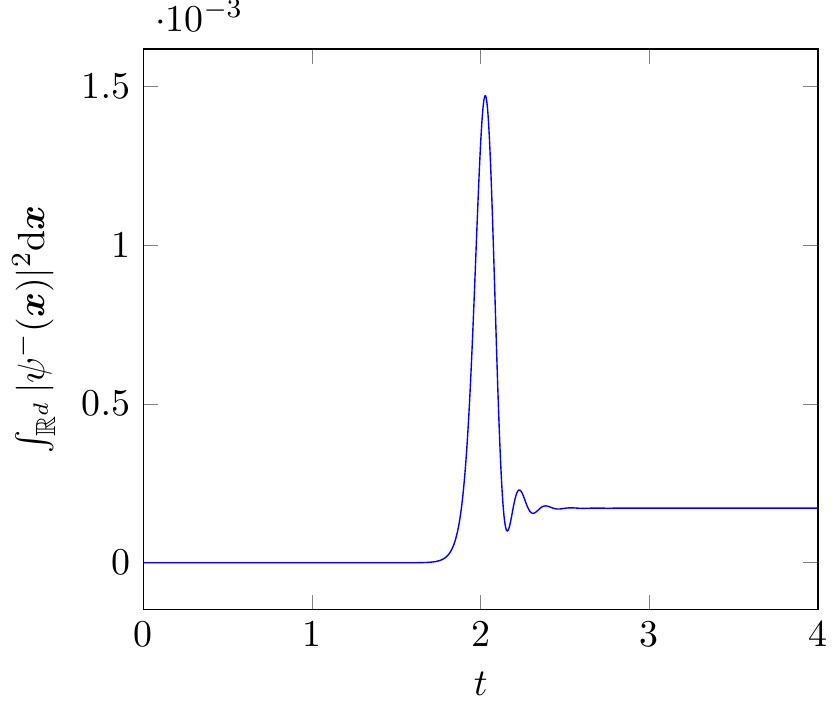}
		\caption{The total mass of wavepacket $\psi^-(\bm{x})$ on the lower potential surface against time $t$, for the system described in \cref{Example:tanh} with parameters in \cref{Parameters1}.
		}\label{fig:MassDown}			
	\end{figure}
	For this reason, the transmitted wavepacket is better approximated using the perturbative solution from the $n^{th}$ {\em superadiabatic} representation \cite{BetzGoddardTeufeul09}, for some optimal choice of $n$. The $n^{th}$ superadiabatic representation is produced by creating and applying unitary {\em pseudodifferential} operators $U_n$, such that the off-diagonal elements of the potential operator have prefactor $\varepsilon^{n+1}$, and the diagonal elements are the same as in the adiabatic representation. Existence of such operators is discussed in \cite{BetzGoddardTeufeul09}. In the $n^{th}$ superadiabatic representation the perturbative solution gives
	\begin{align}
	\psi_n^-(t)=-i\varepsilon^n\int_{-\infty}^te^{-\frac{i}{\varepsilon}(t-s)H_n^-}K_{n+1}^-(\bm{x})e^{-\frac{i}{\varepsilon}sH_n^+}\phi \diff s,\label{eq:SuperAdiabaticPert}
	\end{align} 
	where $H_n$ is the Hamiltonian in the $n^{th}$ superadiabatic representation, given by
	\begin{align}
	H_n=-\frac{\varepsilon^2}{2}\nabla_{\bm{x}}^2\bm{I}+\begin{pmatrix}
	\rho(\bm{x})+d(\bm{x}) & K_{n+1}^+ \\
	K_{n+1}^- & -\rho(\bm{x})+d(\bm{x})
	\end{pmatrix},\label{SuperadiabaticHamiltonian}
	\end{align}	
	for some pseudodifferential coupling operators $K_{n+1}^\pm$, and
	\begin{align}
	H_n^\pm&=-\frac{\varepsilon^2}{2}\nabla_{\bm{x}}^2\pm\rho(\bm{x})+d(\bm{x}).
	\end{align}
	Unfortunately, the need to compute to compute the pseudodifferential operators $K_{n+1}^\pm$  and $U_n$ prevent this from directly 
	producing a practical numerical scheme.  However, as we now demonstrate, we may make use of the superadiabatic representations to obtain a simple and accurate algorithm.
	\section{Approximating the transition in one dimension}\label{Section:1DFormula}
	The derivation in \cite{BetzGoddard09} requires $\rho(\bm{x})\ge\delta>0$ and $\theta,\rho$ to be analytic in a strip containing the real axis. The formula is derived in one dimension using the superadiabatic perturbative solution by
	\begin{enumerate}
		\item Finding algebraic recursive differential equations to calculate the {\em quantum symbol} $\kappa_{n+1}^\pm$, where $K_{n+1}^\pm$ is the {\em Weyl quantisation} of $\kappa_{n+1}^\pm$,
		\begin{align}
		(\mathcal{W}(\kappa_{n+1}^\pm)\psi)(x)&=\frac{1}{2\pi}\int_{\mathbb{R}^{2n}}\diff \xi \diff yK_{n+1}^\pm\left(\xi,\frac{1}{2}(x+y)\right)e^{i\xi\cdot(x-y)}\psi(y).\label{eq:1DWeylQuantisation}
		\end{align}
		\item Rescaling $\kappa_{n+1}^\pm$  by
		\begin{align}
		\tau(q)=2\int_0^q\rho(r)\diff r,\label{eq:1DTau}
		\end{align}
		then approximating $\kappa_{n+1}^\pm$ in an analogous way to the time-adiabatic case in \cite{BetzTeufel05-1}.
		\item Assuming the potential surfaces are approximately linear near the avoided crossing, i.e.\
		$
		H^\pm\approx\frac{\varepsilon^2}{2}\partial_x^2\pm\delta+\lambda x
		$,
		to utilise the Avron-Herbst formula \cite{AvronHerbst77}.
		\item Applying a stationary phase argument to evaluate the remaining integral.
	\end{enumerate}
	The result is presented in {\em scaled momentum space}: in $d$ dimensions the wavepacket in scaled momentum space is given using the {\em $\varepsilon$-scaled Fourier transform} 
	\begin{align}
	\widehat{f}^\varepsilon(\bm{p})=\frac{1}{(2\pi\varepsilon)^{d/2}}\int_{\mathbb{R}^d}f(\bm{x})\exp\left(-\frac{i}{\varepsilon}\bm{p}\cdot \bm{x}\right) \diff \bm{x}.
	\end{align}
	Following this derivation leads to an approximation of the transmitted wavepacket, far from the avoided crossing:
	\begin{align}
	\widehat{\psi^-}^\varepsilon(k,t)&=e^{-\frac{i}{\varepsilon}t\widehat{H}^-(k)}\frac{\nu(k)+k}{2|\nu(k)|}e^{-\frac{\tau_c}{2\delta\varepsilon}|k-\nu(k)|}e^{-\frac{i\tau_r}{2\delta\varepsilon}(k-\nu(k))}\chi_{k^2>4\delta}\widehat{\phi_0^+}^\varepsilon(\nu(k)),\label{eq:1DFormula}
	\end{align}
	where 
	\begin{itemize}
		\item There is no dependence on the $n^{th}$ superadiabatic representation used in the formula derivation.
		\item $\widehat{\phi_0^+}^\varepsilon$ is the wavepacket on the upper level evolved to the avoided crossing using uncoupled dynamics.
		\item $\delta=\text{min}(\rho)$, half the distance between the two adiabatic potential surfaces at the avoided crossing.
		\item  $\nu(k)=\text{sgn}(\sqrt{k^2-4\delta})$, the initial momentum a classical particle would need to have momentum $k$ after falling down a potential energy difference of $2\delta$, {\em i.e.} the distance between the potential surfaces at the avoided crossing, which shifts the wavepacket in momentum space.
		\item $\tau^{cz}=\tau_r + i\tau_c =2\int_0^{q^{cz}}\rho(q)\diff q$, where $q^{cz}$ is the smallest complex zero of $\rho$, when extended to the complex plane. The prefactor $e^{-\frac{\tau_c}{2\delta\varepsilon}|\nu(k)-k|}$ determines the size of the transmitted wavepacket. In \cite{GoddardHurst18}, we show that under appropriate approximations of the momentum and potential surfaces, this prefactor is comparable to the Landau-Zener transition prefactor used in surface hopping algorithms such as in \cite{BelyaevLasserTriglia14}. An additional change in phase occurs due to $\tau_r$, which is present when the potential is not symmetric about the avoided crossing.
	\end{itemize}
	The constructed formula \cref{eq:1DFormula} allows us to approximate the size and shape of the transmitted wave packet due to an avoided crossing, and avoid computing expensive coupled dynamics. The method for applying the algorithm is as follows:
	\begin{enumerate}
		\item Begin with an initial wave packet $\psi_0^+$ on the upper adiabatic energy surface, far from the crossing, with momentum such that the wave packet will cross the minimum of $\rho$ (\cref{fig:1DToyExampleFormulaA}).
		\item Evolve $\psi_0^+$ according to the BOA on the upper adiabatic level until the centre of mass is at the avoided crossing, at time $t^{cz}$ (\cref{fig:1DToyExampleFormulaB}), say $\phi_0^+(x,t^{cz})\coloneqq e^{-\frac{i}{\varepsilon}t^{cz}H^+}\psi_0^+(x)$,
		\item Apply the one dimensional formula to the $\varepsilon$-Fourier transform of the wave packet at the crossing (\cref{fig:1DToyExampleFormulaC}):
		\begin{align}
		\widehat{\psi^-}^\varepsilon(x,t^{cz}) = \frac{\nu(k)+k}{|\nu(k)|}e^{-\frac{\tau_c}{2\delta\varepsilon}|k-\nu(k)|}e^{-\frac{i\tau_r}{2\delta\varepsilon}(k-\nu(k))}\chi_{k^2>4\delta}\widehat{\phi_0^+}^\varepsilon(\nu(k),t^{cz}),
		\end{align}
		\item Evolve the transmitted wave packet far away enough from the crossing, say to time $t^{cz}+s$, using the BOA (\cref{fig:1DToyExampleFormulaD}): $\widehat{\psi^{-}}^\varepsilon(x,t+s)=e^{-\frac{i}{\varepsilon}s\widehat{H}^-}\widehat{\psi^{-}}^\varepsilon(x,t)$. At this time the transmitted wave packet in momentum space should be accurately approximated by the formula result, as long as the wave packet is evolved away from the area in which oscillations occur.
	\end{enumerate}
	\begin{figure}[ht!]
		\centering
		\subfloat[][]{\includegraphics[width=0.5\textwidth]{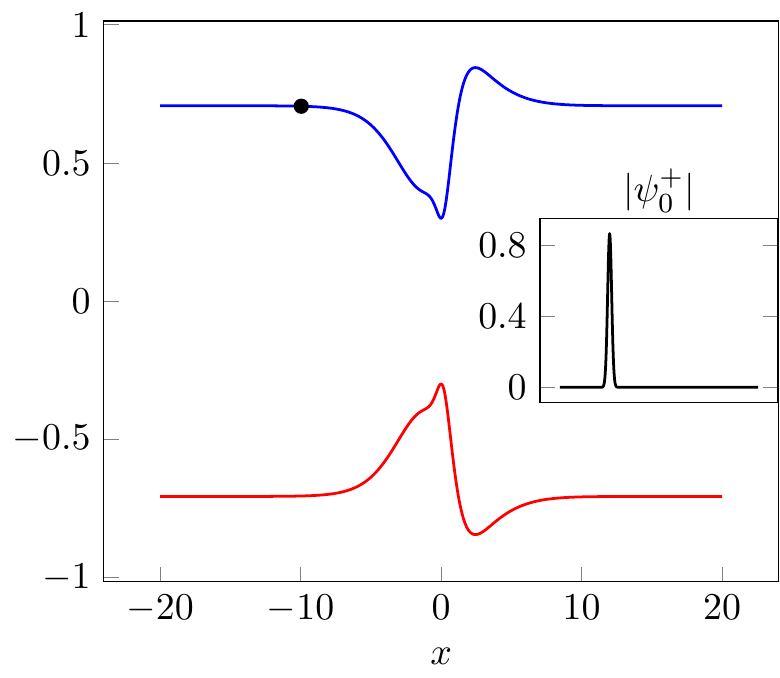}\label{fig:1DToyExampleFormulaA}}
		\subfloat[][]{\includegraphics[width=0.5\textwidth]{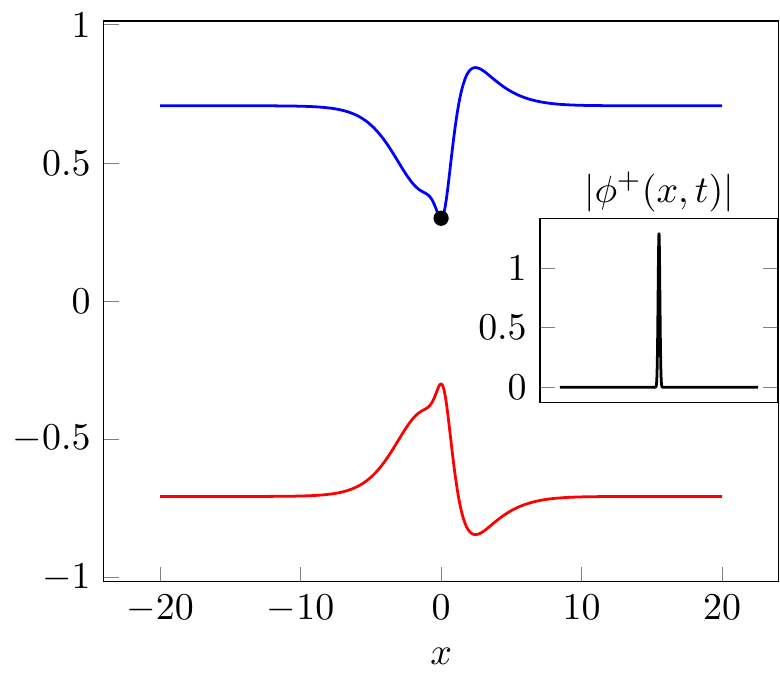}\label{fig:1DToyExampleFormulaB}}\\
		\subfloat[][]{\includegraphics[width=0.5\textwidth]{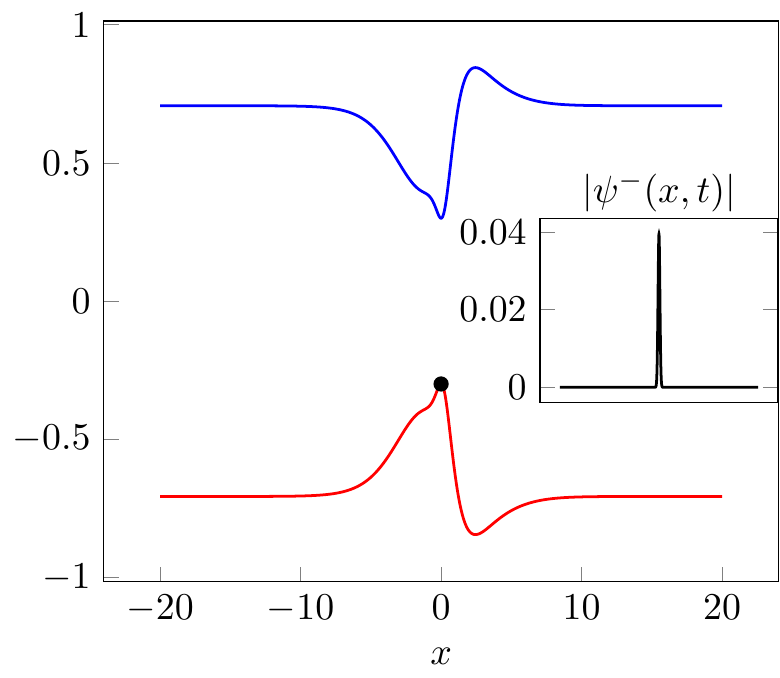}\label{fig:1DToyExampleFormulaC}}
		\subfloat[][]{\includegraphics[width=0.5\textwidth]{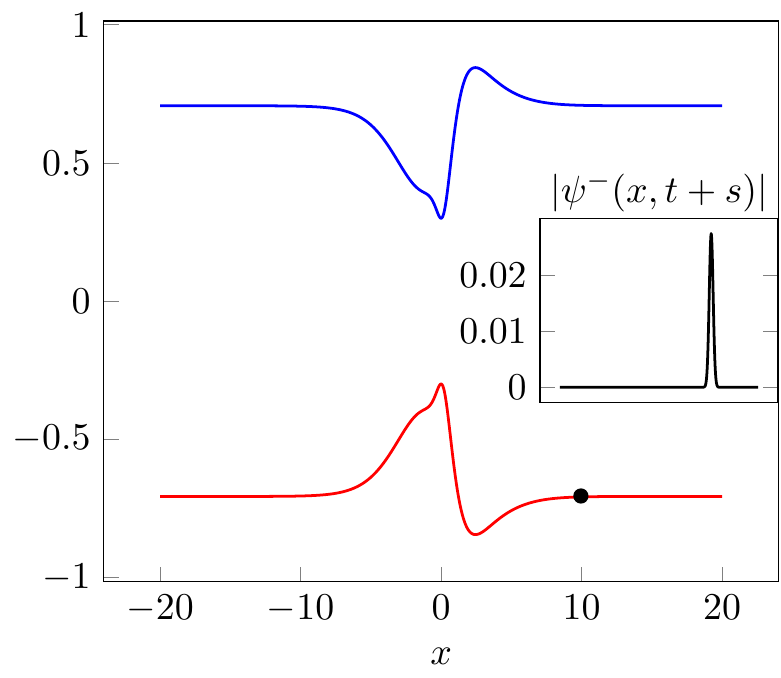}\label{fig:1DToyExampleFormulaD}}
		\caption{Application of the 1D formula for a particular system discussed in \cite{BetzGoddard09}. The centre of mass of the associated wavepacket (inset) is represented by a black point on either the upper (blue) and lower (red) adiabatic potential surfaces.}\label{fig:1DToyExampleFormula}
	\end{figure}
	Note that we have assumed that the avoided crossing is centred at $x=0$. When the avoided crossing is not centred at 0, an additional shift term $e^{\frac{i}{\varepsilon}(\eta(k)-k)x_0}$, where $x_0$ is the position of the avoided crossing, is obtained when changing variables in the Fourier transform of the coupling symbols~\cite{GoddardHurst18}.
	
	Applications of the one dimensional formula have been widely successful on a variety of examples. In \cite{BetzGoddardManthe16}, the formula is used to accurately approximate the transmitted wavepacket for sodium iodide. Tilted avoided crossings have also been examined, and a formula developed which is dependent on the value $n$, so the optimal superadiabatic representation must be calculated. The formula has also been sucessfully applied to model interference effects in multiple transitions \cite{BetzGoddard11}. 
	
	Finally, following the above derivation for reverse transitions (from lower to upper surface), the formula
	\begin{align}
	\widehat{\psi^+}^\varepsilon=-\frac{\tilde{\nu}+k}{2|\tilde{\nu}|}e^{-\tau_c|k-\tilde{\nu}|/(2\delta\varepsilon)}e^{-i\tau_r(k-\tilde{\nu})/(2\delta\varepsilon)}\widehat{\psi_0^+}^\varepsilon(\tilde{\nu}(k)),\label{eq:1DReverseFormula}
	\end{align}
	where $\tilde{\nu}=\text{sgn}(k)\sqrt{k^2+4\delta}$, can be used to approximate the wavepacket transmitted to the upper surface, far from the avoided crossing.
	
	\section{Coupling operators in higher dimensions}\label{Section:NDCouplingOperators}
	The first step in deriving \cref{eq:1DFormula} in \cite{BetzGoddardTeufeul09} was to approximate the superadiabatic coupling operators $K_{n+1}^\pm$. We now consider these operators in higher dimensions. We restrict the calculations here to two dimensions for clarity, but they can easily be adapted to $d$ dimensions.
	\begin{lemma}\label{Lemma:2DRecursions}
		In two dimensions, $\kappa_{n+1}^{\pm}$ is given by
		\begin{align}
		\kappa_{n+1}^\pm(\bm{p},\bm{q})=-2\rho(\bm{q})(x_{n+1}(\bm{p},\bm{q})\pm y_{n+1}(\bm{p},\bm{q})).\label{eq:KappaWithxy}
		\end{align}
		where $x_{n+1}(\bm{p},\bm{q}),y_{n+1}(\bm{p},\bm{q})$ are given by the following algebraic recursive differential equations:
		\begin{align}
		x_1=z_1=w_1=0,\quad y_1=-\frac{i}{4\rho}(p_1\partial_{q_1}\theta+p_2\partial_{q_2}\theta).\label{eq:xyzwStart}
		\end{align}
		and
		\begin{align}
		y_n=0,~n~\text{even},\quad x_n=z_n=w_n=0,~n~\text{odd}.
		\end{align}
		For $n$ odd, we have
		\begin{align}
		x_{n+1}=-\frac{1}{2\rho}\Biggl[\frac{1}{i}
		(\bm{p}\cdot\nabla_{\bm{q}}y_n)-2\sum_{j=1}^n\frac{1}{(2i)^jj!}\sum_{|\bm{\alpha}|=j}\partial_{\bm{p}}^{\bm{\alpha}}(b_{\bm{\alpha}} z_{n+1-j}-a_{\bm{\alpha}} x_{n+1-j})\Biggr],
		\end{align}
		and for $n$ even
		\begin{multline}
		y_{n+1}=-\frac{1}{2\rho}\Bigg[\frac{1}{i}((\bm{p}\cdot\nabla_{\bm{q}}x_n)-z_n(\bm{p}\cdot\nabla_{\bm{q}}\theta))\\
		-2\sum_{j=1}^n\frac{1}{(2i)^jj!}\sum_{|\bm{\alpha}|=j}\partial_{\bm{p}}^{\bm{\alpha}}(-a_{\bm{\alpha}} y_{n+1-j}+b_{\bm{\alpha}} w_{n+1-j})\Bigg],
		\end{multline}
		\begin{align}
		&\frac{1}{i}((\bm{p}\cdot\nabla_{\bm{q}}z_n)-x_n(\bm{p}\cdot\nabla_{\bm{q}}\theta))
		=
		\sum_{j=1}^n\frac{1}{(2i)^jj!}\sum_{|\bm{\alpha}|=j}\partial_{\bm{p}}^{\bm{\alpha}}(b_{\bm{\alpha}} y_{n+1-j}+a_{\bm{\alpha}} w_{n+1-j}),\\
		&\frac{1}{i}(\bm{p}\cdot\nabla_{\bm{q}}w_n)
		=
		2\sum_{j=1}^n\frac{1}{(2i)^jj!}\sum_{|\bm{\alpha}|=j}\partial_{\bm{p}}^{\bm{\alpha}}(a_{\bm{\alpha}} z_{n+1-j}+b_{\bm{\alpha}} x_{n+1-j}),\label{eq:xyzwEnd}
		\end{align}
		where $\bm{\alpha} = (\alpha_1,\alpha_2)$ and $\partial_{\bm{p}}^{\bm{\alpha}}=\partial_{p_1}^{\alpha_1}\partial_{p_2}^{\alpha_2}$, and $a_\alpha$ and $b_\alpha$ are given by the recursions
		\begin{gather*}
		a_0=\rho(q_1,q_2),\quad b_0=0,\nonumber\\
		a_{(\alpha_1+1,\alpha_2)}=\partial_{q_1}a_{(\alpha_1,\alpha_2)}+(\partial_{q_1}\theta)b_{(\alpha_1,\alpha_2)},\quad
		b_{(\alpha_1+1,\alpha_2)}=\partial_{q_1}b_{(\alpha_1,\alpha_2)}-(\partial_{q_1}\theta)a_{(\alpha_1,\alpha_2)},\nonumber\\
		a_{(\alpha_1,\alpha_2+1)}=\partial_{q_2}a_{(\alpha_1,\alpha_2)}+(\partial_{q_2}\theta)b_{(\alpha_1,\alpha_2)},\quad
		b_{(\alpha_1,\alpha_2+1)}=\partial_{q_2}b_{(\alpha_1,\alpha_2)}-(\partial_{q_2}\theta)a_{(\alpha_1,\alpha_2)}.
		\end{gather*}
	\end{lemma}
	\begin{proof}
		The method is a straightforward extension of the theory in \cite{BetzGoddardTeufeul09}.
	\end{proof}
	As in \cite{BetzGoddardTeufeul09}, by general theory, the coefficients $x_n,y_n,z_n,w_n$ are polynomials in $\bm{p}$ of order $n$. We therefore write
	\begin{align}
	x_n(\bm{p},\bm{q})=\sum_{m=0}^n\sum_{k=0}^mp_1^kp_2^{m-k}x_n^{k,m-k} (q_1,q_2),\label{eq:xyzwExpansion}
	\end{align}
	for some $x_n^{k,m-k}(q_1,q_2)$, and similarly for $y_n,z_n,w_n$. For a given $j$, we write $\bm{\alpha}_j = (\alpha,j-\alpha)$ for each $\alpha\le j$.
	Note that:
	\begin{align}
	\partial_{\bm{p}}^{\bm{\alpha}_j}p_1^kp_2^{m-k}&=
	\begin{cases}
	\frac{k!}{(m-k)!}\frac{(m-k)!}{(m-k-j+\alpha)!}p_1^{k-\alpha}p_2^{m-k-j+\alpha},& k\ge\alpha\text{ and }m\ge j,\\
	0,&\text{otherwise}
	\end{cases}
	\end{align}
	Then
	\begin{align*}
	\partial_{\bm{p}}^{\bm{\alpha}_j}x_{n+1-j}&=\sum_{m=0}^{n+1-j}\sum_{k=0}^m(\partial_{p_1}^\alpha p_1^k)\big(\partial_{p_2}^{j-\alpha}p_2^{m-k}\big)x_{n+1-j}^{k,m-k}(q_1,q_2),\\
	&= \sum_{m=j}^{n+1-j}\sum_{k=\alpha}^{m-\alpha+j}
	\frac{k!}{(k-\alpha)!}\frac{(m-k)!}{(m-k-j+\alpha)!}p_1^{k-\alpha}p_2^{m-k-j+\alpha}
	x_{n+1-j}^{k,m-k}(q_1,q_2),
	\end{align*}
	so that
	\begin{align}
	\mathcal{A}\coloneqq
	\sum_{j=1}^n\frac{1}{(2i)^jj!}\sum_{\alpha=0}^j\partial_{p_1}^\alpha\partial_{p_2}^{j-\alpha}a_{\bm{\alpha}_j}x_{n+1-j},\label{eq:derivativexn+1-j}
	\end{align}
	can be rewritten as
	\begin{align*}
	\sum_{j=1}^n\frac{1}{(2i)^jj!}\sum_{\alpha=0}^ja_{\bm{\alpha}_j}\sum_{m=j}^{n+1-j}\sum_{k=\alpha}^{m+\alpha-j}\frac{k!}{(k-\alpha)!}\frac{(m-k)!}{(m-k-j+\alpha)!}p_1^{k-\alpha}p_2^{m-k-j+\alpha}
	x_{n+1-j}^{k,m-k}.
	\end{align*}
	We now want to extract $p_1$ and $p_2$ from the final two summations, 
	so that we can compare coefficients on either side of the results of \cref{Lemma:2DRecursions} to construct recursive equations for $x_{n}^{A,B}$ for $A+B<n$. Consider terms where $j>\frac{n+1}{2}$. By the limits of the third summand, we find that $m>\frac{n+1}{2}$, and that $m<\frac{n+1}{2}$, a contradiction. Therefore we restrict the limits of first summand.
	Defining $b=k-\alpha$, and		
	$c=m-j$, we find
	\begin{align*}
	\mathcal{A}&=\sum_{j=1}^{\floor{\frac{n+1}{2}}}\sum_{\alpha=0}^j\sum_{c=0}^{n+1-2j}\sum_{b=0}^{c}\frac{a_{\bm{\alpha}_j}}{(2i)^jj!}\frac{(b+\alpha)!}{b!}\frac{((c+j)-(b+\alpha))!}{(c-b)!}p_1^{b}p_2^{c-b}
	x_{n+1-j}^{b+\alpha,(c+j)-(b+\alpha)}.
	\end{align*}
	We now want to switch the order of summation. We note that, for an arbitrary $\mathcal{B}$,
	\begin{align*}
	\sum_{j=1}^{\floor{\frac{n+1}{2}}}\sum_{c=0}^{n+1-2j}\mathcal{B}_{c,j}=\sum_{c=0}^{n+1}\sum_{j=1}^{\floor{\frac{c}{2}}}\mathcal{B}_{n+1-c,j},
	\end{align*}
	which can be shown directly (note that the terms where $c=0$, $c=1$ are zero). Using this, we finally have that
	\begin{multline}
	\mathcal{A}=\sum_{c=0}^{n+1}\sum_{b=0}^{n+1-c}p_1^{b}p_2^{n+1-c-b}\\
	\times \sum_{j=1}^{\floor{\frac{c}{2}}}
	\sum_{\alpha=0}^j\frac{a_{\bm{\alpha}_j}}{(2i)^jj!}\frac{(b+\alpha)!}{b!}\frac{(n+1-c+j-b-\alpha)!}{(n+1-c-b)!}
	x_{n+1-j}^{b+\alpha,(n+1-c+j)-(b+\alpha)}.\label{eq:ASum}
	\end{multline}
	Importantly, $p_1$ and $p_2$ have been extracted from two of the summations. Note that taking $b=0$ and $\alpha=0$, or $j-\alpha=0$ and $n+1-c-b=0$ in \cref{eq:ASum}, we return the 1D result in \cite{BetzGoddardTeufeul09} for $p_2$ and $p_1$ respectively. We obtain the following result.
	\begin{proposition}\label{Prop:2DCoeffs}
		The coefficients $x_n^{A,B}$ to $w_n^{A,B}$ are determined by the following algebraic-differential recursive equations. We have
		\begin{gather}
		x_1^{A,B}=z_1^{A,B}=w_1^{A,B}=0,\quad A+B\in \{0,1\},\\
		y_1^{0,0}=y_1^{1,1}=0,\quad y_1^{1,0}=-\frac{i}{4\rho}\partial_{q_1}\theta,\quad y_1^{0,1}=-\frac{i}{4\rho}\partial_{q_2}\theta.\label{eq:xyzwExpStart}
		\end{gather}
		Further,
		\begin{multline}
		x_{n+1}^{A,B}=-\frac{1}{2\rho}\Bigg[
		\frac{1}{i}(\partial_{q_1}y_n^{A-1,B}+\partial_{q_2}y_n^{A,B-1})
		-2\sum_{j=1}^{\floor{\frac{n+1-(A+B)}{2}}}\sum_{\alpha=0}^{j}\frac{1}{(2i)^jj!}\\
		\times\frac{(A+\alpha)!}{A!}
		\frac{(B+j-\alpha)!}{B!}
		\left(b_{\bm{\alpha}_j}z_{n+1-j}^{A+\alpha,B+j-\alpha}-a_{\bm{\alpha}_j}x_{n+1-j}^{A+\alpha,B+j-\alpha}\right)
		\Bigg],
		\end{multline}
		When $n$ is odd. When $n$ is even, we have
		\begin{multline}
		y_{n+1}^{A,B}=-\frac{1}{2\rho}\Bigg[
		\frac{1}{i}((\partial_{q_1}x_n^{A-1,B}+\partial_{q_2}x_n^{A,B-1})-(z_n^{A-1,B}\partial_{q_1}\theta+z_n^{A,B-1}\partial_{q_2}\theta))\\
		-2\sum_{j=1}^{\floor{\frac{n+1-(A+B)}{2}}}\sum_{\alpha=0}^j\frac{1}{(2i)^jj!}\frac{(A+\alpha)!}{A!}\frac{(B+j-\alpha)!}{B!}\\
		\times\left(-a_{\bm{\alpha}_j}y_{n+1-j}^{A+\alpha,B+j-\alpha}+b_{\bm{\alpha}_j}w_{n+1-j}^{A+\alpha,B+j-\alpha}\right)\Bigg],
		\end{multline}
		\begin{multline}
		0=\frac{1}{i}((\partial_{q_1}z_n^{A-1,B}+\partial_{q_2}z_n^{A,B-1})+(x_n^{A-1,B}\partial_{q_1}\theta+x_n^{A,B-1}\partial_{q_2}\theta))\\
		-2\sum_{j=1}^{\floor{\frac{n+1-(A+B)}{2}}}\sum_{\alpha=0}^j\frac{1}{(2i)^jj!}\frac{(A+\alpha)!}{A!}\frac{(B+j-\alpha)!}{B!}\\
		\times\left(b_{\bm{\alpha}_j}y_{n+1-j}^{A+\alpha,B+j-\alpha}+a_{\bm{\alpha}_j}w_{n+1-j}^{A+\alpha,B+j-\alpha}\right),
		\end{multline}
		\begin{multline}
		0=\frac{1}{i}((\partial_{q_1}w_n^{A-1,B}+\partial_{q_2}w_n^{A,B-1})\\
		-2\sum_{j=1}^{\floor{\frac{n+1-(A+B)}{2}}}\sum_{\alpha=0}^j\frac{1}{(2i)^jj!}\frac{(A+\alpha)!}{A!}\frac{(B+j-\alpha)!}{B!}\\
		\times\left(a_{\bm{\alpha}_j}z_{n+1-j}^{A+\alpha,B+j-\alpha}+b_{\bm{\alpha}_j}x_{n+1-j}^{A+\alpha,B+j-\alpha}\right).\label{eq:xyzwExpEnd}
		\end{multline}
	\end{proposition}
	\begin{proof}
		We substitute \cref{eq:xyzwExpansion} into the results of \cref{Lemma:2DRecursions} and compare coefficients in powers of $p_1,p_2$ on either side, using \cref{eq:ASum}.
	\end{proof}
	As with the coefficients $x_n$ and $y_n$ in \cref{eq:KappaWithxy}, $\kappa_{n+1}^\pm$ has polynomial form:
	\begin{align}
	\kappa_{n+1}^{\pm}(\bm{p},\bm{q})=\sum_{m=0}^n\sum_{j=0}^mp_1^jp_2^{m-j}\kappa_{n+1}^{(j,m-j)\pm} (q_1,q_2).\label{eq:KappaExpansion}
	\end{align}
	The coupling operators only act on wavepackets near the avoided crossing, so when considering the effect of the coupling operator on a wavepacket we assume that the path of the wavepacket is approximately linear. Furthermore, by a change of coordinate system, we may assume that the direction of travel of the wavepacket is independent of $p_2$ (i.e.\ we rotate the frame of reference so that the wavepacket is moving in the $p_1$ direction). We also assume that the leading order term in $p_1$ dominates $\kappa_{n+1}^\pm$: $\kappa_{n+1}^{\pm}\approx p_1^n\kappa_{n+1}^{(n,0)\pm}(q_1,q_2)$. This is similar to the assumption made in the one dimensional case, where it can be shown to be accurate for sufficiently large $p$, but in practice holds for much smaller values. Then the 2D algebraic differential recursive equations then reduce to the one dimensional case in \cite{BetzGoddardTeufeul09}:
	\begin{gather}
	x^{n+1,0}_{n+1}\approx\frac{i}{2\rho}(\partial_{q_1}y^{n,0}_n),\nonumber \\ y^{n+1,0}_{n+1}\approx\frac{i}{2\rho}((\partial_{q_1}x^{n,0}_n)'-(\partial_{q_1}\theta) z_n^{n,0}),\quad 0\approx\partial_{q_1}z_n^{n,0}+(\partial_{q_1}\theta) x_n^{n,0}.\label{eq:xyzwLeading}
	\end{gather}
	To ease notation, redefine $x_{n+1}=x_{n+1}^{n+1,0}$, and similar for $y_{n+1},z_{n+1}$. It is unclear what the analogue of \cref{eq:1DTau}, introduced initially in \cite{BerryLim93} for the time-adiabatic case, would be for multidimensional systems. We introduce the natural scaling in the first dimension
	\begin{align}
	\tau(q_1,q_2)=2\int_0^{q_1}\rho(r,q_2)\diff r. \label{eq:nDTauApprox1D}
	\end{align}
	Defining $\tilde{f}(\tau(q_1,q_2))=f(q_1,q_2)$ the recursive relations \cref{eq:xyzwLeading} then become
	\begin{align}
	\tilde{x}_{n+1}^0=i\tilde{y}_{n+1}^0,\quad \tilde{y}_{n+1}^0=i((\tilde{x}_{n}^0)'+\tilde{\theta}'\tilde{z}_n^0),\quad
	0=(\tilde{z}_n^0)'+\tilde{\theta}'\tilde{x}_n^0,\label{eq:xyzwTau}
	\end{align}
	where $\tilde{\theta}'=\frac{d}{d\tau(q_1,q_2)}\tilde{\theta}$. These recursive equations also occur in \cite{BetzTeufel05-1}, where they are solved in one dimension, under the assumption that
	\begin{align}
	\frac{d}{d\tau}\tilde{\theta}(\tau)&=\frac{i\gamma}{\tau-\bar{\tau}^{cz}}-\frac{i\gamma}{\tau-\tau^{cz}}+\tilde{\theta}_r'(\tau),\label{eq:ThetaAssumption}
	\end{align}
	where $\tau^{cz}$ is a first order complex singularity of $\tilde{\theta}$, and $\tilde{\theta}_r$ has no singularities closer to the real axis than $\tau^{cz}$. Assuming that the avoided crossing occurs at $0$, we can write $\rho^2(q)=\delta^2+g(q)^2$, for some analytic function $g$ such that $g(0)\approx0$, and $g^2$ is quadratic in the neighbourhood of $q=0$. Therefore a Stokes line ({\em i.e.}\ a curve with Im$(\rho)=0$) crosses the real axis perpendicularly \cite{JoyeMiletiPfister91}, and following this line leads to a pair of complex conjugate points $q^{cz},\bar{q}^{cz}$ which are complex zeros of $\rho$. Defining $\tau^{cz}=\tau(q^{cz})$, it is shown in \cite{BerryLim93} that first order complex singularities of the adiabatic coupling function arise at these complex zeros. This derivation is still valid in our case, for each $q_2$. The recursive algebraic differential equations solved in \cite{BetzTeufel05-1} then give us $\kappa_{n,0}^-$ to leading order:
	\begin{align}
	\kappa_{n}^-(\bm{q})\approx\frac{i^n}{\pi}\rho(\bm{q})(n-1)!\left(\frac{i}{(\tau(\bm{q})-\bar{\tau}^{cz}(q_2))^n}-\frac{i}{(\tau(\bm{q})-\tau^{cz}(q_2))^{n}}\right).\label{eq:Kappan0}
	\end{align}
	It is clear that the results of this section can be extended to higher dimensions, by assuming the direction of travel of the wavepacket is in the first dimension. We will now use this observation to design an algorithm for multi-dimensional transitions using only the 1D transition formula.
	\section{The multi-dimensional algorithm}\label{Section:SlicingAlgorithm}
	The derivation of a multidimensional formula, under the assumptions above, follow similarly to the one dimensional case. First we define the multidimensional Weyl quantization \cite{Bach02}.
	\begin{definition}
		For a symbol $H(\varepsilon,\bm{p},\bm{q})$, given a test function $\psi$, we define the {\em Weyl quantization} of $H$ by
		\begin{align}
		(\mathcal{W}_\varepsilon H\psi)(\bm{x})=\frac{1}{(2\pi\varepsilon)^{d}}\int_{\mathbb{R}^{2d}}\diff \bm{\xi}\diff \bm{y} H(\varepsilon,\bm{\xi},\frac{1}{2}(\bm{x}+\bm{y}))e^{\frac{i}{\varepsilon}(\bm{\xi}\cdot(\bm{x}-\bm{y}))}\psi(\bm{y}).\label{eq:nDWeyl}
		\end{align}
	\end{definition}
	We want to approximate the pseudodifferential operator $K_n$, which is given by the Weyl quantisation of $\kappa_n$. The particular form of $\kappa_n$ allows us to simplify the Weyl quantisation as follows.  
	\begin{proposition}\label{Proposition:nDWeylKappa}
		Let $\kappa(\bm{p},\bm{q})=g(\bm{q})\prod_{i=1}^dp_i^{A_i}$, for $A_i\in\mathbb{N}$. Then
		\begin{align}
		\widehat{(\mathcal{W}_\varepsilon\kappa)(\psi)}^\varepsilon(\bm{k})=\frac{1}{(2\pi\varepsilon)^{d/2}}\int_{\mathbb{R}^d}
		\widehat{g}^\varepsilon(\bm{k}-\bm{\eta})
		\prod_{i=1}^d\left(\frac{k_i+\eta_i}{2}\right)^{A_i}\widehat{\psi}^\varepsilon(\bm{\eta})\diff \bm{\eta}.\label{eq:WeylKappaResult}
		\end{align}
	\end{proposition}
	\begin{proof}
		Firstly, using that $\psi(\bm{y}) = (2\pi\varepsilon)^{-d/2} \int_{\mathbb{R}^d} \diff \bm{\eta} \widehat{\psi}^\varepsilon(\bm{\eta}) \exp(i (\bm{\eta}\cdot\bm{y})/\varepsilon)$,
		\begin{align*}
		\mathcal{W}_\varepsilon\kappa\psi(\bm{x})=&\frac{1}{(2\pi\varepsilon)^d}\int_{\mathbb{R}^{2d}}\diff \bm{\xi}\diff \bm{y}
		\left(\prod_{i=1}^d\xi_i^{A_i}\right)
		g\left(\frac{\bm{x}+\bm{y}}{2}\right)e^{\frac{i}{\varepsilon}(\bm{\xi}\cdot(\bm{x}-\bm{y}))}\psi(\bm{y}),\\
		=&\frac{1}{(2\pi\varepsilon)^{3d/2}}\int_{\mathbb{R}^{3d}}d\bm{\xi}\diff \bm{y} \diff\bm{\eta}
		\left(\prod_{i=1}^d\xi_i^{A_i}\right)
		g\left(\frac{\bm{x}+\bm{y}}{2}\right)e^{\frac{i}{\varepsilon}(\bm{\xi}\cdot(\bm{x}-\bm{y})+\bm{\eta}\cdot\bm{y})}\widehat{\psi}^\varepsilon(\bm{\eta}).
		\end{align*}
		Now define $\tilde{y}_i= (x_i+y_i)/2$, $i=1,...,d$. Then
		\begin{align*}
		\mathcal{W}_\varepsilon\kappa\psi(\bm{x})&=\frac{2^d}{(2\pi\varepsilon)^{3d/2}}\int_{\mathbb{R}^{3d}}\diff \bm{\xi} \diff \bm{\tilde y}\diff \bm{\eta}
		\left(\prod_{i=1}^d\xi_i^{A_i}\right)
		g(\tilde{\bm{y}})
		e^{\frac{i}{\varepsilon}(\bm{\xi}\cdot\bm{x}+(2\tilde{\bm{y}}-\bm{x})\cdot(\bm{\eta}-\bm{\xi}))}\psi(\bm{\eta}),\\
		&=\frac{2^d}{(2\pi\varepsilon)^d}\int_{\mathbb{R}^{2d}}\diff \bm{\xi}\diff \bm{\eta}
		\left(\prod_{i=1}^d\xi_i^{A_i}\right)
		e^{\frac{i}{\varepsilon}(\bm{x}\cdot(2\bm{\xi}-\bm{\eta}))}\psi(\bm{\eta})
		\hat{g}^\varepsilon(2(\bm{\xi}-\bm{\eta})).
		\end{align*}
		We perform a second change of variables
		$\tilde{\xi}_i=2\xi_i$
		and find
		\begin{align*}
		\mathcal{W}_\varepsilon\kappa\psi(\bm{x})&=\frac{1}{(2\pi\varepsilon)^d}\int_{\mathbb{R}^{2d}}\diff \tilde{\bm{\xi}}\diff \bm{\eta}
		\left(\prod_{i=1}^d\left(\frac{\tilde{\xi}_i}{2}\right)^{A_i}\right)
		e^{\frac{i}{\varepsilon}(\bm{x}\cdot(\tilde{\bm{\xi}}-\bm{\eta}))}\hat{\psi}^\varepsilon(\bm{\eta})\hat{g}^\varepsilon(\tilde{\bm{\xi}}-2\bm{\eta}).
		\end{align*}
		We apply the scaled Fourier transform to both sides of this equation:
		\begin{align*}
		\widehat{\mathcal{W}_\varepsilon\kappa\psi}^\varepsilon(\bm{k})=\frac{1}{(2\pi\varepsilon)^{3d/2}}\int_{\mathbb{R}^{3d}}& \diff \bm{\tilde \xi} \diff \bm{\eta} \diff \bm{x} \left(\prod_{i=1}^d\left(\frac{\tilde{\xi}_i}{2}\right)^{A_i}\right)
		e^{\frac{i}{\varepsilon}(\bm{x}\cdot(\tilde{\bm{\xi}}-\bm{\eta}-\bm{k}))}\hat{\psi}^\varepsilon(\bm{\eta})\hat{g}^\varepsilon(\bm{\tilde \xi}-2\bm{\eta}).
		\end{align*}
		Using that $(2\pi \varepsilon)^{-d} \int \diff \bm{x} \exp(i (\bm{a}\cdot \bm{x})/\varepsilon)= \delta(\bm{a})$ allows us to directly
		compute the $\bm{x}$ integral, giving \cref{eq:WeylKappaResult}.
	\end{proof}
	Next we linearise the dynamics near the avoided crossing. To leading order the uncoupled propagators in \ref{eq:AdiaTerms} can be approximated by 
	$
	H^\pm_1=-\frac{\varepsilon^2}{2}\nabla_{\bm{x}}^2\pm\delta+\bm{\lambda}\cdot\bm{x}
	$.
	Then, by the fundamental theorem of calculus,
	\begin{align}
	e^{\frac{i}{\varepsilon}sH^\pm}-e^{\frac{i}{\varepsilon}sH_1^\pm}&=e^{-\frac{i}{\varepsilon}sH_1^\pm}\int_0^se^{\frac{i}{\varepsilon}rH_1^\pm}\left[\frac{i}{\varepsilon}(H_1^\pm-H^\pm)\right]e^{-\frac{i}{\varepsilon}rH^\pm}\diff r.\label{eq:LineariseIntegrals}
	\end{align}
	Since $H_1^\pm-H^\pm$ is quadratic near zero, the integrand in \cref{eq:LineariseIntegrals} is of order 1 in an $\sqrt{\varepsilon}$-neighbourhood of zero. Outside of this region the coupling function provides a negligible result, as seen in the one dimensional case \cite{BetzGoddardTeufeul09}. We also use the $d$ dimensional Avron-Herbst formula \cite{AvronHerbst77}, which shows that
	\begin{align}
	e^{-\frac{i}{\varepsilon}s\widehat{H_1^\pm}\varepsilon}&=e^{-\frac{i\|\bm{\lambda}\|^2s^3}{6\varepsilon}}e^{s(\bm{\lambda}\cdot\partial_{\bm{k}})}e^{-\frac{i}{2\varepsilon}((\|\bm{k}\|^2\pm2\delta)s-(\bm{\lambda}\cdot\bm{k})s^2)}.\label{eq:nDAvronHerbst}
	\end{align}
	Then
	\begin{align}
	\widehat{\psi^-_n}^\varepsilon(\bm{k},t)&\approx
	-i\varepsilon^ne^{-\frac{i}{\varepsilon}t\widehat{H^-}^\varepsilon}
	\int_{-\infty}^{t}
	e^{-\frac{i\|\bm{\lambda}\|^2s^3}{6\varepsilon}}e^{s(\bm{\lambda}\cdot\partial_{\bm{k}})}e^{-\frac{i}{2\varepsilon}((\|\bm{k}\|^2-2\delta)s-(\bm{\lambda}\cdot\bm{k})s^2)}
	\widehat{K_{n+1}^-}^\varepsilon\nonumber\\
	&\qquad\qquad\qquad\qquad \times e^{-\frac{i\|\bm{\lambda}\|^2s^3}{6\varepsilon}}e^{s(\bm{\lambda}\cdot\partial_{\bm{k}})}e^{-\frac{i}{2\varepsilon}((\|\bm{k}\|^2+2\delta)s-(\bm{\lambda}\cdot\bm{k})s^2)}
	\widehat{\phi_0^+}^\varepsilon(\bm{k})\diff s.
	\end{align}
	Using \cref{Proposition:nDWeylKappa} for the coupling function shows that
	\begin{multline*}
	\widehat{\psi^-_n}^\varepsilon(\bm{k},t)\approx-i\frac{\varepsilon^n}{(2\pi\varepsilon)^{d/2}}e^{-\frac{i}{\varepsilon}t\widehat{H^-}^\varepsilon}\int_{-\infty}^{t}\diff s
	e^{-\frac{i\|\bm{\lambda}\|^2s^3}{6\varepsilon}}e^{s(\bm{\lambda}\cdot\partial_{\bm{k}})}e^{-\frac{i}{2\varepsilon}((\|\bm{k}\|^2-2\delta)s-(\bm{\lambda}\cdot\bm{k})s^2)}
	\\
	\times\int_{\mathbb{R}^d}\diff \bm{\eta}
	\left\{\sum_{A_i=1,i=1,..,d}^{n+1}\widehat{\kappa_{n+1}^{\bm{A},-}}^\varepsilon(\bm{k}-\bm{\eta})
	\left(\prod_{i=1}^d\left(\frac{k_i+\eta_i}{2}\right)^{A_i}\right)
	\right\}\\
	\times e^{-\frac{i(\|\bm{\lambda}\|^2s^3}{6\varepsilon}}e^{s(\bm{\lambda}\cdot\partial_{\bm{\eta}})}e^{-\frac{i}{2\varepsilon}((\|\bm{\eta}\|^2+2\delta)s-(\bm{\lambda}\cdot\bm{\eta})s^2)}\widehat{\phi_0^+}^\varepsilon(\bm{\eta_1}),
	\end{multline*}
	where $\bm{A}=(A_1...A_d)$. The operator $e^{s\bm{\lambda}\cdot\partial_{\bm{k}}}$ is a {\em shift operator}, so $e^{s\bm{\lambda}\cdot\partial_{\bm{k}}}f(\bm{k})=f(\bm{k}+\bm{\lambda} s)$. Instead of applying the shift operator to the right, we use the fact that the integral is invariant under the transform $\bm{\eta}\mapsto\bm{\eta}-\bm{\lambda} s$ to apply it to the left: in this case $f(\bm{\eta})e^{-s\bm{\lambda}\cdot\partial_{\bm{\eta}}}=f(\bm{\eta}-\bm{\lambda} s)$. The following transformations take place in the integrand:
	\begin{gather*}
	\widehat{\kappa_{n+1}^{\bm{A},-}}^\varepsilon(\bm{k}-\bm{\eta})\mapsto \widehat{\kappa_{n+1}^{\bm{A},-}}^\varepsilon(\bm{k}-\bm{\eta}),\quad
	\bm{k}+\bm{\eta}\mapsto \bm{k}+\bm{\eta}-2\bm{\lambda} s,\\
	e^{\frac{i}{2\varepsilon}((\|\bm{k}\|^2\pm2\delta)s-(\bm{\lambda}\cdot\bm{k})s^2)}\mapsto
	e^{\frac{i}{2\varepsilon}((\|\bm{k}-\bm{\lambda} s\|^2\pm2\delta)s-(\bm{\lambda}\cdot(\bm{k}-\bm{\lambda} s))s^2)}.
	\end{gather*}
	Rearranging gives
	\begin{multline}
	\widehat{\psi^-_n}^\varepsilon(\bm{k},t)\approx
	-i\frac{\varepsilon^n}{(2\pi\varepsilon)^{d/2}}e^{-\frac{i}{\varepsilon}t\widehat{H^-}^\varepsilon}
	\\
	\times\int_{-\infty}^{t}\int_{\mathbb{R}^d}\diff s\diff \bm{\eta}
	\left\{
	\sum_{A,B=1}^{n+1}\widehat{\kappa_{n+1}^{\bm{A},-}}^\varepsilon(\bm{k}-\bm{\eta})
	\left(\prod_{i=1}^d\left(\frac{k_i+\eta_i-2\lambda_i s}{2}\right)^{A_i}\right)
	\right\}
	\\
	\times\widehat{\phi_0^+}^\varepsilon(\bm{\eta})
	\exp\left\{\frac{i}{2\varepsilon}\left[(\|\bm{k}\|^2-\|\bm{\eta}\|^2-4\delta)s-(\bm{\lambda}\cdot(\bm{k}-\bm{\eta}))s^2\right] \right\}.\label{eq:psiTestminuseq}
	\end{multline}	
	We approximate $\kappa_{n+1}^-$ with \cref{eq:Kappan0}, then calculate the scaled Fourier transform:
	\begin{align*}
	&\widehat{\kappa_{n}^-}^\varepsilon(\bm{k})=\frac{1}{(2\pi\varepsilon)^{d/2}}\int_{\mathbb{R}^d}e^{-(i/\varepsilon)\bm{k}\cdot\bm{q}}\kappa_{n}^-(\bm{q})\diff \bm{q},\\
	&\approx\frac{(n-1)!}{(2\pi\varepsilon)^{d/2}}\frac{i^n}{\pi}\int_{\mathbb{R}^d}\rho(\bm{q})\left[\frac{i}{(\tau(\bm{q})-\bar{\tau}^{cz}(\bm{q}^{d-1}))^n}-\frac{i}{(\tau(\bm{q})-\tau^{cz}(\bm{q}^{d-1}))^n}\right]e^{-(i/\varepsilon)\bm{k}\cdot\bm{q}}\diff \bm{q},
	\end{align*}
	where $\bm{q}^{d-1}=(q_2,...,q_d)$. We now assume that $\rho$ does not vary significantly in any dimension other than the first. Then $\rho(\bm{q})\approx\rho(q_1)$, and consequently $\tau(\bm{q})=\tau(q_1),\tau^{cz}(\bm{q}^{d-1})=\tau^{cz}$. Therefore the Fourier transform in all other dimensions is given by $\frac{1}{\sqrt{2\pi\varepsilon}}\int_{-\infty}^{\infty}e^{-\frac{ikx}{\varepsilon}}\diff x=\sqrt{2\pi\varepsilon}\delta(k)$.
	As $\tau(\bm{q})\approx\tau(q_1)$, we only need to consider the one dimensional case. This is discussed in \cite{BetzGoddardTeufeul09}. A simple extension to $d$ dimensions therefore shows that
	\begin{align}
	\widehat{\kappa_{n,0}^-}^\varepsilon(\bm{k})&=
	\frac{i}{\sqrt{2\pi\varepsilon}}
	\left(\frac{k_1}{2\delta\varepsilon}\right)^{n-1}
	e^{-i\tau_r\frac{k_1}{2\delta\varepsilon}}
	e^{-\tau_c\frac{|k_1|}{2\delta\varepsilon}}
	\sqrt{2\pi\varepsilon}^{(d-1)}
	\delta(k_2,...,k_d).\label{eq:kappaHat}
	\end{align}
	We insert \cref{eq:kappaHat} into \cref{eq:psiTestminuseq}, and rearrange to find
	\begin{multline}
	\widehat{\psi_n^-}(\bm{k},t)=
	\frac{1}{4\pi\varepsilon}e^{-\frac{i}{\varepsilon}t\widehat{H^-}^\varepsilon}
	\int_0^\infty \diff s
	\int_{\mathbb{R}} \diff \eta_1
	\left(\frac{k_1^2-\eta_1^2}{4\delta}\right)^n\left(1-\frac{2\lambda_1 s}{k_1+\eta_1}\right)^{n+1}\\
	\times e^{-\frac{i\tau_r(k_1-\eta_1)}{2\delta\varepsilon}}e^{-\frac{\tau_c(|k_1-\eta_1|)}{2\delta\varepsilon}}\nonumber\\
	\times\left\{
	\int_{\mathbb{R}^{d-1}}\diff \eta_2...\diff \eta_d
	\widehat{\phi_0^+}^\varepsilon(\bm{\eta})e^{\frac{i}{2\varepsilon}[(\|\bm{k}\|^2-\|\bm{\eta}\|^2)s-\bm{\lambda}\cdot(\bm{k}-\bm{\eta})s^2]}
	\delta(k_2-\eta_2,...,k_d-\eta_d)
	\right\}.
	\end{multline}	
	By the identity $f(x)=\int_{-\infty}^{\infty}\delta(x-a)f(a)\diff a$,
	the integral in the dimensions $2,...,d$ can be evaluated to find
	\begin{multline}
	\widehat{\psi_n^-}(\bm{k},t)=
	\frac{1}{4\pi\varepsilon}e^{-\frac{i}{\varepsilon}t\widehat{H^-}^\varepsilon}
	\int_0^\infty \diff s
	\int_{\mathbb{R}} \diff \eta_1
	\left(\frac{k_1^2-\eta_1^2}{4\delta}\right)^n\left(1-\frac{2\lambda_1 s}{k_1+\eta_1}\right)^{n+1}\\
	\times e^{-\frac{i\tau_r(k_1-\eta_1)}{2\delta\varepsilon}}e^{-\frac{\tau_c(|k_1-\eta_1|)}{2\delta\varepsilon}}\nonumber\\
	\times\widehat{\phi_0^+}^\varepsilon(\eta_1,k_2,...,k_d)e^{\frac{i}{2\varepsilon}[(|k_1|^2-|\eta_1|^2-4\delta)s-\lambda_1(k_1-\eta_1)s^2]}.
	\end{multline}
	We assume that $\lambda_1$ is small and so can be neglected, so that
	\begin{align}
	\widehat{\psi_n^-}(\bm{k},t)=&
	\frac{1}{4\pi\varepsilon}e^{-\frac{i}{\varepsilon}t\widehat{H^-}^\varepsilon}
	\int_0^\infty \diff s
	\int_{\mathbb{R}} \diff \eta
	\left(\frac{k_1^2-\eta_1^2}{4\delta}\right)^n
	e^{-\frac{i\tau_r(k_1-\eta_1)}{2\delta\varepsilon}}e^{-\frac{\tau_c(|k_1-\eta_1|)}{2\delta\varepsilon}}\nonumber\\&\times
	\widehat{\phi_0^+}^\varepsilon(\eta_1,k_2,...,k_d)e^{\frac{i}{2\varepsilon}(|k_1|^2-|\eta_1|^2-4\delta)s}.
	\end{align}
	Although here we restrict to flat avoided crossings, we expect the result for one dimensional tilted avoided crossings \cite{BetzGoddard11}, when $\lambda_1\neq0$, should also be applicable in higher dimensions. From here we can follow the derivation in \cite{BetzGoddardTeufeul09}	
	and obtain an extension its main result to $d$ dimensions, under the assumptions that the direction of travel is in the first dimension, and that $\rho$ does not vary significantly in other directions:
	\begin{multline}
	\widehat{\psi^-}^\varepsilon(\bm{k},t)=e^{-\frac{i}{\varepsilon}t\widehat{H}^-(\bm{k})}\frac{\nu(k_1)+k_1}{2|\nu(k_1)|}	e^{-\frac{i}{\varepsilon}(k_1-\nu(k_1))x_{0}}
	e^{-\frac{\tau_c}{2\delta\varepsilon}|k_1-\nu(k_1)|}\\
	\times e^{-\frac{i\tau_r}{2\delta\varepsilon}(k_1-\nu(k_1))}\chi_{k_1^2>4\delta}\widehat{\phi_0^+}^\varepsilon(\nu(k_1),k_2,...,k_d),\label{eq:nDFlatFormula}
	\end{multline}
	By the linearity of the Schr\"{o}dinger equation this result can be extended to systems where $\rho$ does vary in other directions by decomposing space into strips and approximating the potential on each strip. We outline the method with the following algorithm and 2D diagrams in \cref{fig:SlicingAlgorithm}:
	\begin{enumerate}\label{Algorithm:Slicing}
	
		\item Begin with an initial wave packet $\psi_0^+(\bm{x})$ on the upper adiabatic energy surface, far from the crossing, with momentum such that the centre of mass of the wavepacket will obtain a minimum value of $\rho$ (\cref{fig:SlicingAlgorithmA}).
		
		\item Evolve $\psi_0^+$ on the upper level, i.e.\ under the BOA, until its centre of mass reaches a local minimum at time $t$: $\phi_0^+(\bm{x})\coloneqq e^{-\frac{i}{\varepsilon}tH^+}\psi_0^+(\bm{x})$.
		
		\item Calculate the centre of momentum $\bm{p}_{\mathrm{COM}}=\frac{\int_{\mathbb{R}^n}\diff \bm{p}\bm{p}|\widehat{\phi_0^+}^\varepsilon(\bm{p})|^2}{\int_{\mathbb{R}^n}\diff \bm{p}|\widehat{\phi_0^+}^\varepsilon(\bm{p})|^2}$ of $\phi_0^+(\bm{x})$.
		
		\item Divide up the full $d$-dimensional space into $d$-dimensional strips parallel to $\bm{p}_{\mathrm{COM}}$.  
		The width of the strips in all directions perpendicular to $\bm{p}_{\mathrm{COM}}$ 
		should be of the order of the width of the transition region (along $\bm{p}_{\mathrm{COM}}$) 
		in the optimal superadiabatic basis.
		In practice we restrict these strips to the region of space where the wavepacket has significant mass.
		
		\item On each strip, replace the true potential energy matrix by an approximation that is flat perpendicular to the direction of $\bm{p}_{\mathrm{COM}}$.  In practice, we take the potential along $\bm{p}_{\mathrm{COM}}$ and replicate it in the 
		directions perpendicular to $\bm{p}_{\mathrm{COM}}$.  Note in particular that the new potential may be different for each strip.
				
		\item Compute the transmitted wavepacket on the lower level for each strip by applying the formula \cref{eq:nDFlatFormula}  along $\bm{p}_{\mathrm{COM}}$ (\cref{fig:SlicingAlgorithmD}) and sum them together:
		$\widehat{\psi^-}^\varepsilon(\bm{k},t)=\sum_{j=1}^n\widehat{\psi_j^-}^\varepsilon(\bm{k},t)$.
				
		\item Evolve the transmitted wavepacket away from the avoided crossing on the lower level, say to time $t+s$, using the BOA (\cref{fig:SlicingAlgorithmF}): $\widehat{\psi^{-}}^\varepsilon(\bm{k},t+s)=e^{-\frac{i}{\varepsilon}s\widehat{H^-}^\varepsilon}\widehat{\psi^{-}}^\varepsilon(\bm{k},t)$.
	\end{enumerate}

As justification for the proposed algorithm we note that we are evolving the wavepacket on the new potential energy surface, 
restricted to each  strip.  As such, we discard any part of the wavepacket that leaves the strip and ignore any additional parts entering from 
other strips. Since the Schr\"odinger equation is linear, this introduces two types of error, due to:
(i) the modification of the potential in each strip, and (ii) the wavepacket broadening out of the selected strip, or into it from the outside.
Both errors are small, the first because the strip is quite narrow (so the potential is approximately constant), the second because the time 
that we actually evolve for is small (of the order of the crossing region in the optimal superadiabatic basis).

In practice, for the examples in \cref{Section:Numerics}, we compute the  BOA dynamics on a uniform 2-dimensional grid.  Once the centre of mass
 of the wavepacket reaches the crossing, we interpolate the wavepacket onto a grid with the new $p_1$ direction 
 parallel to that of $\bm{p}_{\mathrm{COM}}$.  
 Instead of treating strips of the appropriate width, we simply apply the formula  \cref{eq:nDFlatFormula} along each of the 1D lines
 parallel to $p_1$ (or $\bm{p}_{\mathrm{COM}}$); this reduces to applying the 1D formula.  
 For small $\epsilon$, this is essentially equivalent to the algorithm above as the approximate potentials
 of neighbouring lines are very similar and the evolution time in the optimal superadiabatic basis is very short.

To summarise, we have derived an algorithm for approximating the transmitted wavepacket for an avoided crossing in any dimension, which only requires one-level dynamics, and local information about the adiabatic electronic surfaces, {\em i.e.}\ $\delta$ and $\tau^{cz}$. A similar method can be used to determine transmitted wavepackets from lower to upper levels. In the following section, we show that the algorithm above can accurately and efficiently produce an approximation for the transmitted wavepacket.

	\section{Numerical results}\label{Section:Numerics}
	We perform the algorithm on a selection of examples, and compare it to the two level `exact' computation, where the Strang Splitting method is used. For all examples we consider two wavepackets given in momentum space by:
	\begin{align}
	\widehat{\psi_0}^\varepsilon(\bm{p}) &= \frac{1}{N_\psi}\exp{\left(-\frac{\|\bm{p}-\bm{p}_0\|^2}{2\varepsilon}\right)}\exp{\left(-i\frac{(\bm{p}-\bm{p}_0)\cdot\bm{x}_0}{\varepsilon}\right)}\label{eqn:Gaussian},\\
	\widehat{\phi_0}^\varepsilon(\bm{p}) &= \frac{1}{N_\phi}\exp{\left(-\frac{\|\bm{p}-\bm{p}_0\|^6}{2\varepsilon}\right)}\exp{\left(-i\frac{(\bm{p}-\bm{p}_0)\cdot\bm{x}_0}{\varepsilon}\right)}\label{eqn:NonGaussian},
	\end{align}
	where $N_\alpha$ are normalisation constants. To know the momentum of the wavepacket at the avoided crossing, we choose to define the wavepackets at the avoided crossing point, then evolve backwards in time away from the avoided crossing using one level dynamics, before evolving forwards and applying the formula. In practice the initial wavepacket can be given in any initial location, provided it is far enough from the avoided crossing to be unaffected by coupling effects.
	
	To compare the formula results to exact calculations we use the $L^2$-relative error:
	\begin{align}
	Er_{\mathrm{rel}}(\psi_1,\psi_2) = \max\left(\frac{\|\psi_1\pm\psi_2\|}{\|\psi_1\|},\frac{\|\psi_1\pm\psi_1\|}{\|\psi_2\|}\right),
	\end{align}
	Where $\|\cdot\|$ is the standard $L^2$-norm. For comparison to other algorithms which do not calculate phase, it is also beneficial to consider the relative absolute error
	\begin{align}
	Er_{\mathrm{abs}}(\psi_1,\psi_2) = \max\left(\frac{\||\psi_1|-|\psi_2|\|}{\|\psi_1\|},\frac{\||\psi|_1-|\psi|_1\|}{\|\psi_2\|}\right).
	\end{align}
	or the relative mass error
	\begin{align}
	Er_{\mathrm{mass}}(\psi_1,\psi_2) = \max\left(\frac{\|\psi_1\|}{\|\psi_2\|},\frac{\|\psi_2\|}{\|\psi_1\|}\right)-1.
	\end{align}
	\begin{example}\label{Example:tanh}
		Consider the diabatic potential matrix
		\begin{align}
		V(\bm{x}) = \begin{pmatrix}
		\tanh(x_1) & \delta \\
		\delta & -\tanh(x_1) 
		\end{pmatrix}.
		\end{align}
		This is a direct extension of a one dimensional problem, and as there is no dependence in $x_2$, the assumptions made in the derivation in \cref{Section:SlicingAlgorithm} are exactly valid, if the direction of the wavepacket is independent of $p_2$. The lower surface is given by $V_L=-V_U$. The upper adiabatic surface is shown in \cref{fig:VUTanh}.
		\begin{figure}[ht!]
			\centering
			\subfloat[][]{\includegraphics[width=0.5\textwidth]{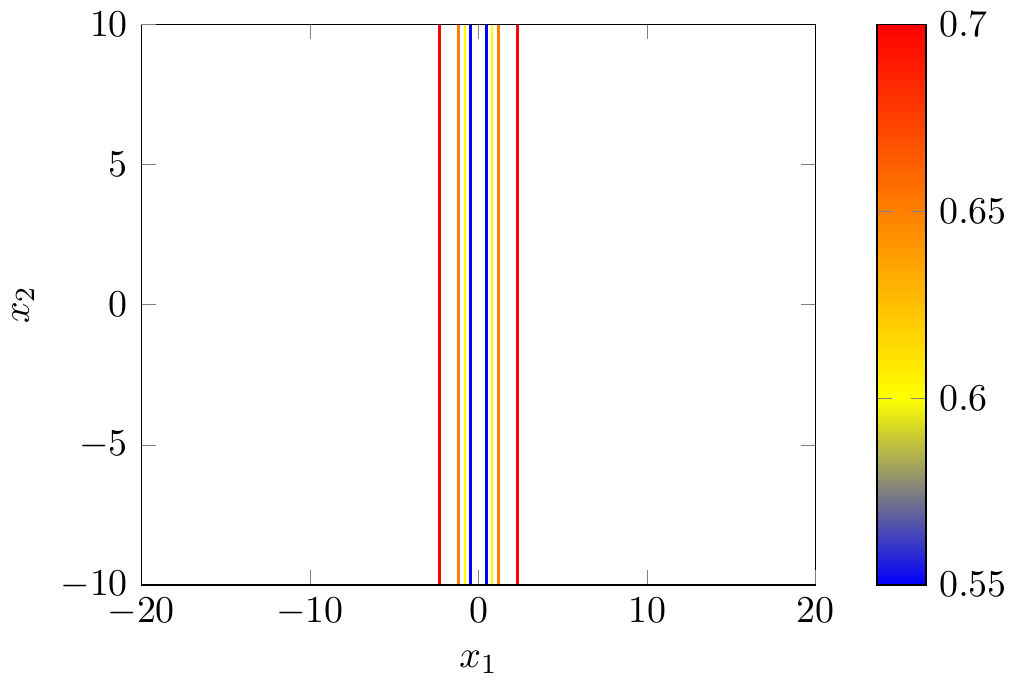}\label{fig:VUTanh}}
			\subfloat[][]{\includegraphics[width=0.48\textwidth]{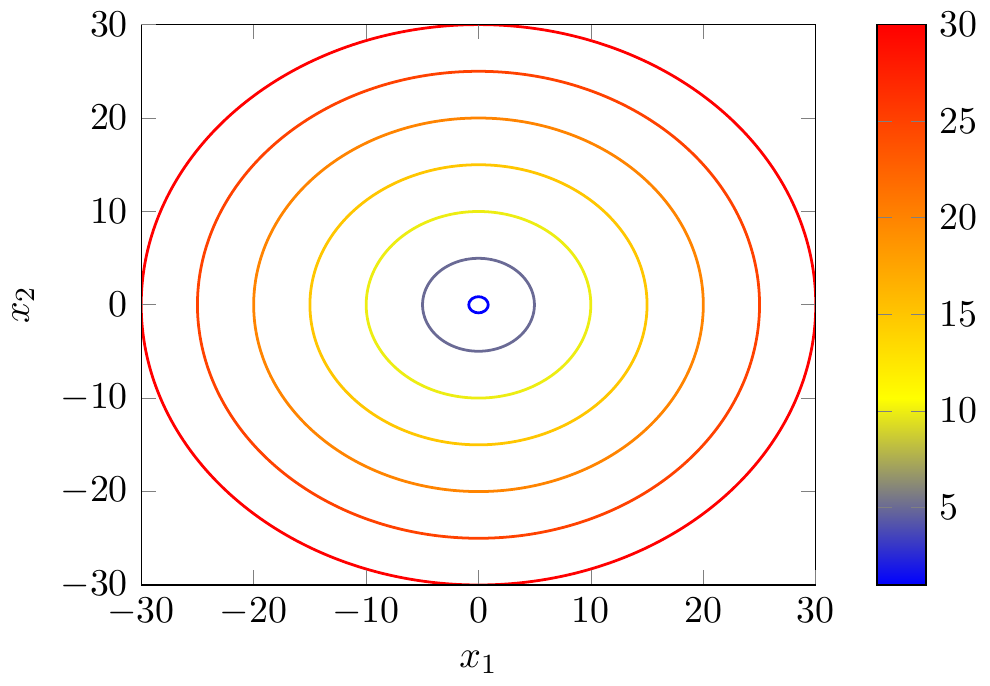}\label{fig:VUJT}}\\
			\caption{Contour plot of the upper adiabatic potential surfaces for \cref{Example:tanh} (left) and \cref{Example:CJLM} (right). In these examples, $V_U=-V_L$. }\label{fig:Potentialsurfaces}			
		\end{figure}
		We take parameters
		\begin{align}\label{Parameters1}
		\left\{\varepsilon,\delta,\bm{p}_0,\bm{x_0}\right\} &= \left\{\frac{1}{30},\frac{1}{2},(6,1),(0,0)\right\}.
		\end{align}
		Using a mesh of $2^{13}\times 2^{13}$ points on the domain $[-20,20]^2$, starting at time 0, we evolve the wavepacket back to time -2 with time-step $1/(50*\|\bm{p}_0\|)$, then evolve forwards to time 2, applying the algorithm, and compare to the exact calculation. For the Gaussian wavepacket $\psi$, $Er_{\mathrm{rel}} = 0.0151$, $Er_{\mathrm{abs}}= 0.0151$, and $Er_{\mathrm{mass}}= 0.0016$. For non-Gaussian $\phi$ $Er_{\mathrm{rel}} = 0.0389$, $Er_{\mathrm{abs}}=0.0387$, and $Er_{\mathrm{mass}}= 0.0023$. The result of the formula and corresponding error are shown in \cref{fig:tanhFormulaResult,fig:tanhFormulaResultNG}. 
		\begin{figure}[ht!]
			\centering
			\subfloat[][]{\includegraphics[width=0.5\textwidth]{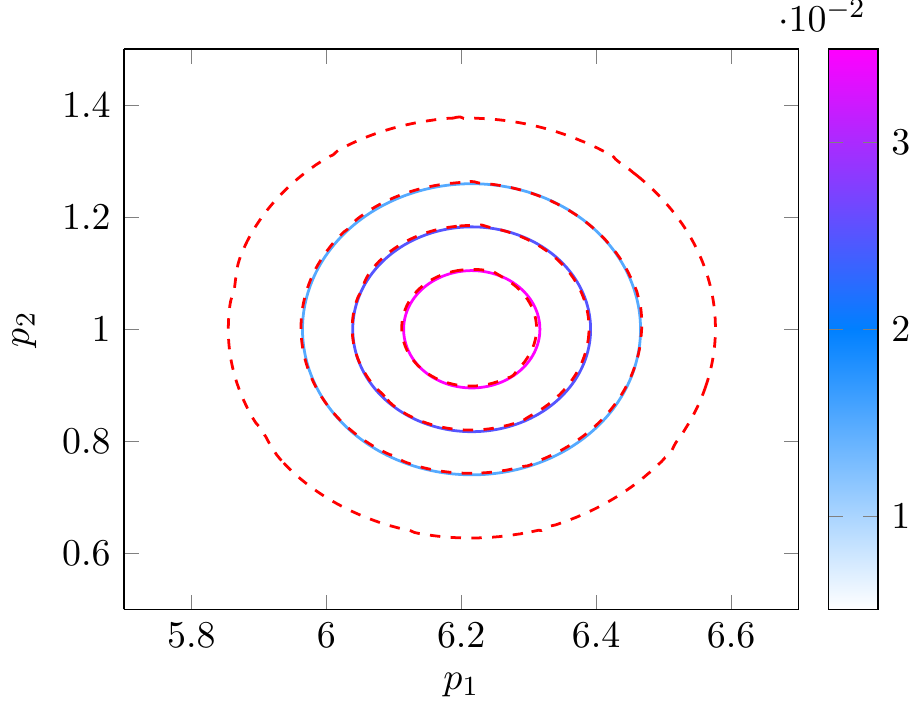}\label{fig:tanhFormulaResult1}}
			\subfloat[][]{\includegraphics[width=0.5\textwidth]{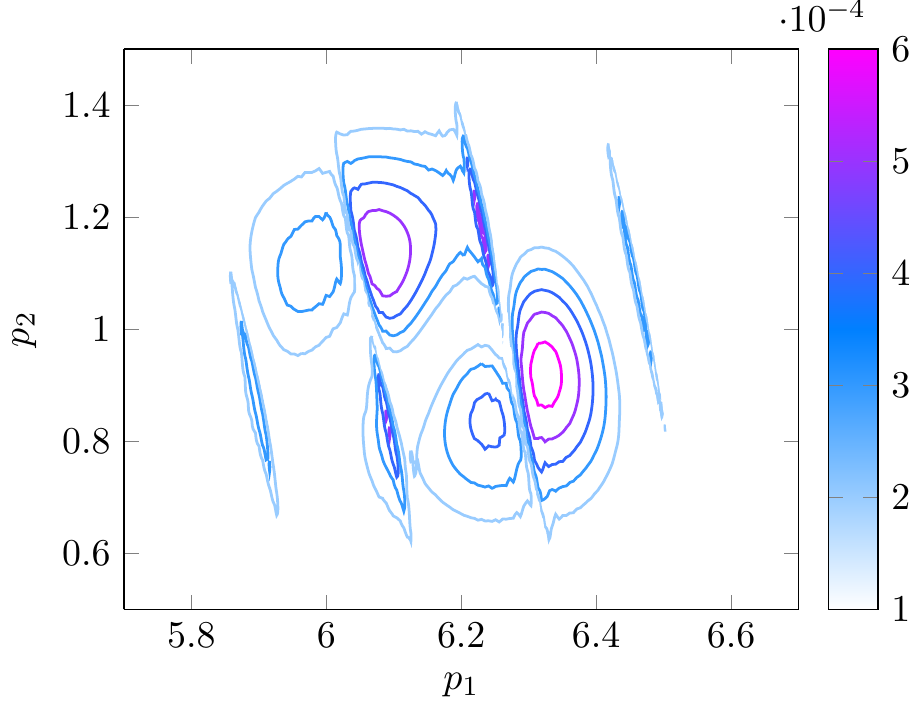}\label{fig:tanhFormulaResult2}}\\
			\caption{Results for \cref{Example:tanh}, when using parameters in \cref{Parameters1} with initial wavepacket of form \cref{eqn:Gaussian}. Left: exact calculation (solid line) versus formula result (dashed line). Contours for the formula result are at the same values as the neighbouring exact contours. Right: relative error.}\label{fig:tanhFormulaResult}			
		\end{figure}
		\begin{figure}[ht!]
			\centering
			\subfloat[][]{\includegraphics[width=0.485\textwidth]{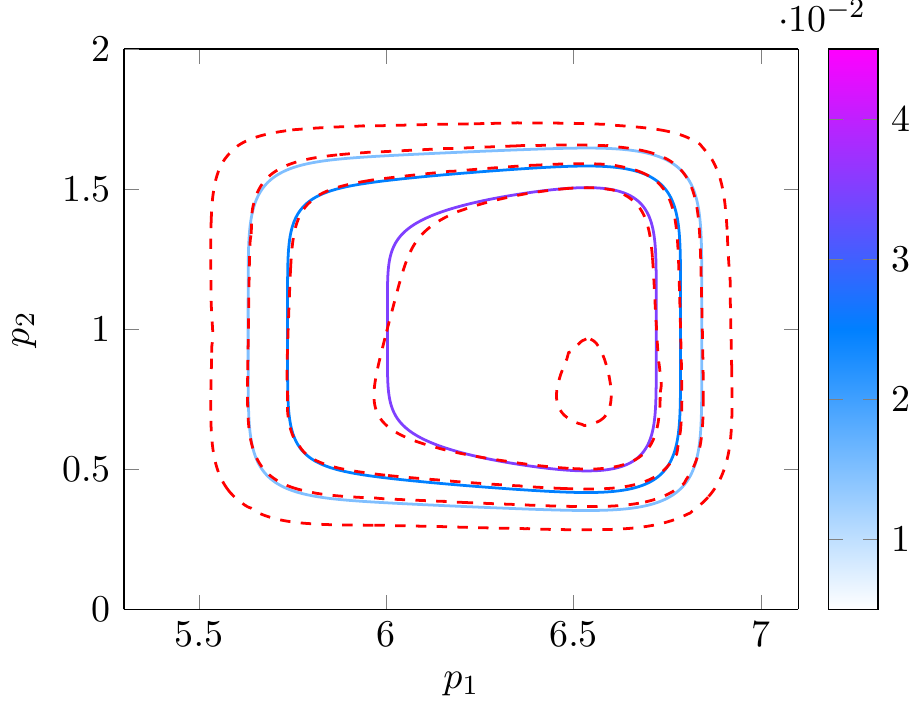}\label{fig:tanhFormulaResult1NG}}
			\subfloat[][]{\includegraphics[width=0.5\textwidth]{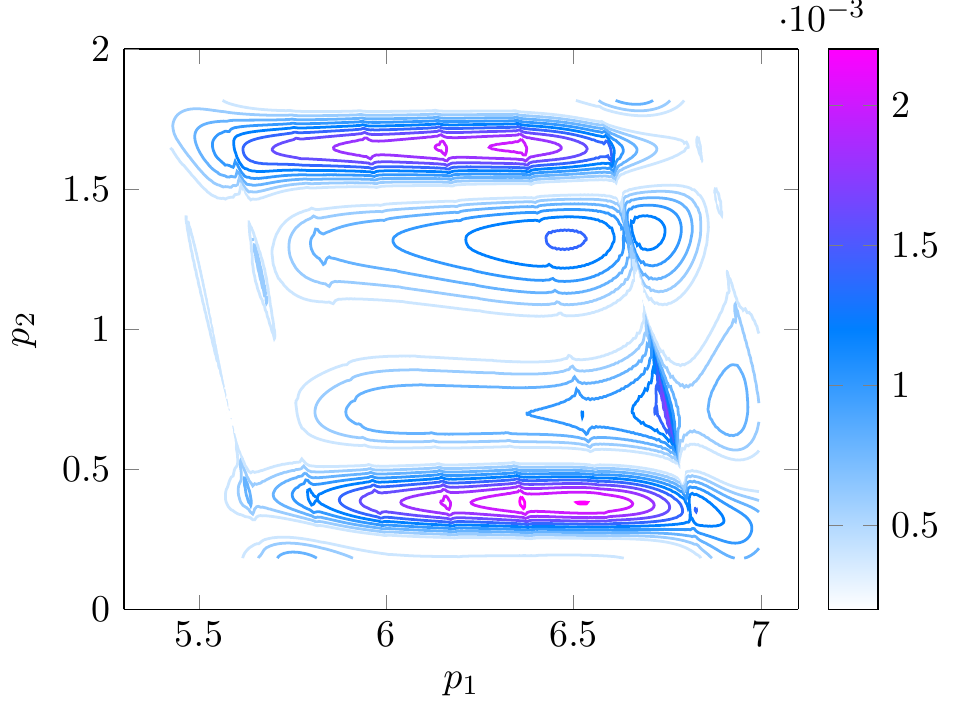}\label{fig:tanhFormulaResult2NG}}\\
			\caption{As in \cref{fig:tanhFormulaResult}, but with initial wavepacket  \cref{eqn:NonGaussian}.}\label{fig:tanhFormulaResultNG}			
		\end{figure}
	\end{example}
	\begin{example}\label{Example:CJLM}
		We consider the diabatic potential matrix described in \cite{ChaiJinLiMorandi15}
		\begin{align}
		V(\bm{x}) = \begin{pmatrix}
		x_1 & \sqrt{x_2^2+\delta^2} \\
		\sqrt{x_2^2+\delta^2} & -x_1 
		\end{pmatrix},
		\end{align}
		which is a modified Jahn-Teller diabatic potential, where the conical intersection is replaced with an avoided crossing with gap $2\delta$. The upper adiabatic surface is shown in \cref{fig:VUJT}. We use parameters
		\begin{align}\label{Parameters2}
		\left\{\varepsilon,\delta,\bm{p}_0,\bm{x_0}\right\} &= \left\{\frac{1}{30},0.5,(5,2),(0,0)\right\},
		\end{align}	
		a mesh of $2^{13}\times 2^{13}$ points on the domain $[-40,40]^2$, we start at time 0, and evolve backwards with time-step $1/(50*\|\bm{p}_0\|)$ to time $-20/\|\bm{p}_0\|^2$, then forwards to $20/\|\bm{p}_0\|^2$, we find $Er_{\mathrm{rel}} = 0.0351, Er_{\mathrm{abs}}=0.0304 $, and $Er_{\mathrm{mass}}=0.0029$ using Gaussian initial wavepacket $\psi_0$, and $Er_{\mathrm{rel}} = 0.0679, Er_{\mathrm{abs}}= 0.0616$, and $Er_{\mathrm{mass}}=0.0033$ for non-Gaussian initial wavepacket $\phi_0$. \Cref{fig:CJLMFormulaResult,fig:CJLMFormulaResultNG} display the result of the formula compared to the exact calculation. 
		\begin{figure}[ht!]
			\centering
			\subfloat[][]{\includegraphics[width=0.5\textwidth]{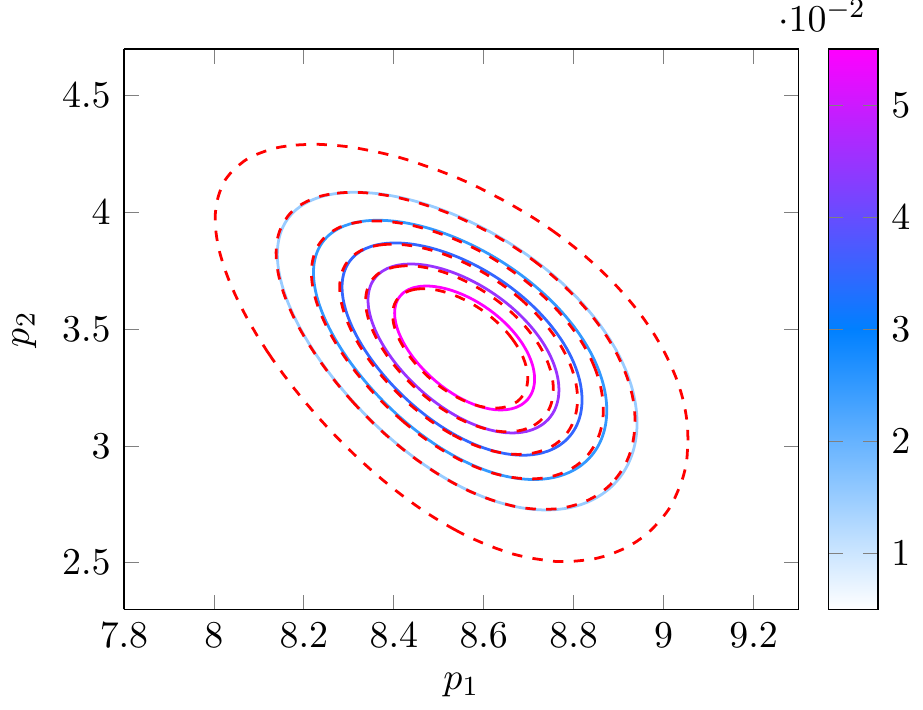}\label{fig:CJLMFormulaResult1}}
			\subfloat[][]{\includegraphics[width=0.5\textwidth]{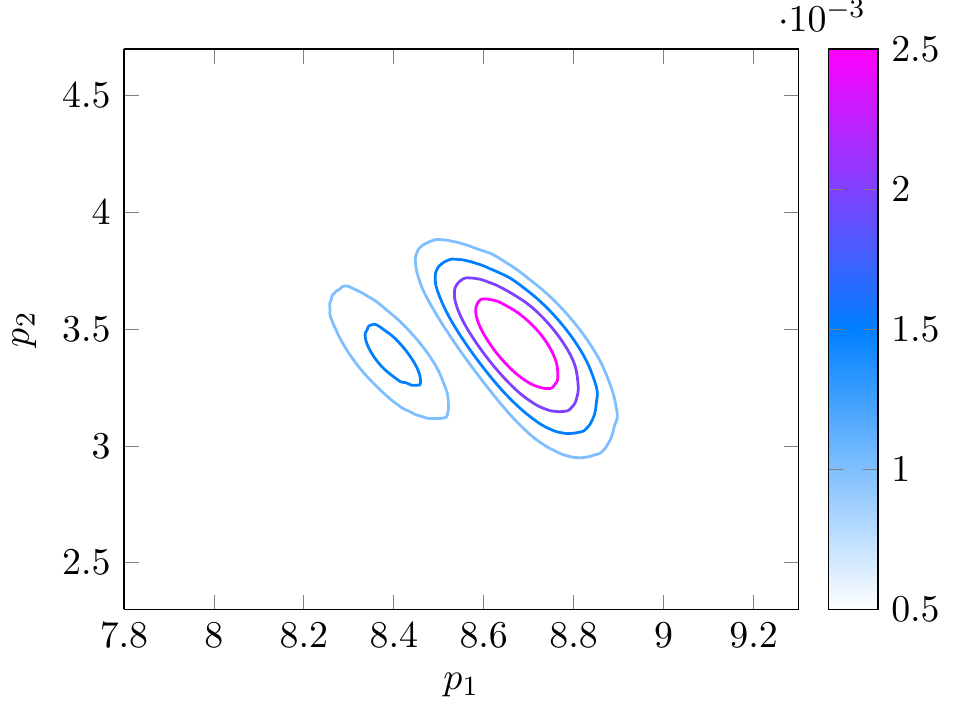}\label{fig:CJLMFormulaResult2}}\\
			\caption{Results for \cref{Example:CJLM}, when using parameters in \cref{Parameters2} with initial wavepackets of form \cref{eqn:Gaussian}. Results are presented as in \cref{fig:tanhFormulaResult}.}\label{fig:CJLMFormulaResult}			
		\end{figure}
		\begin{figure}[ht!]
			\centering
			\subfloat[][]{\includegraphics[width=0.48\textwidth]{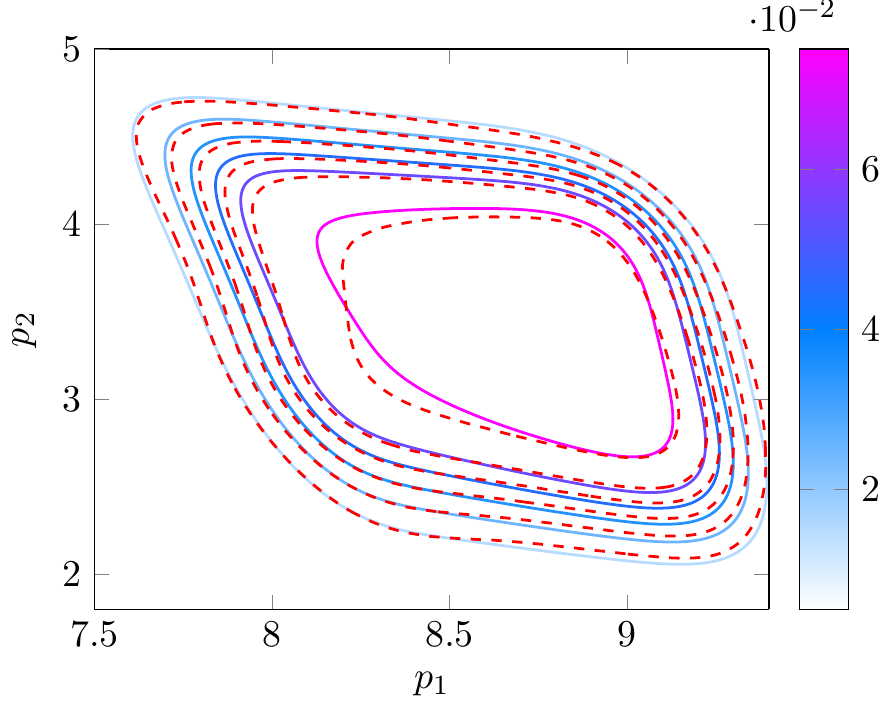}\label{fig:CJLMFormulaResult1NG}}
			\subfloat[][]{\includegraphics[width=0.5\textwidth]{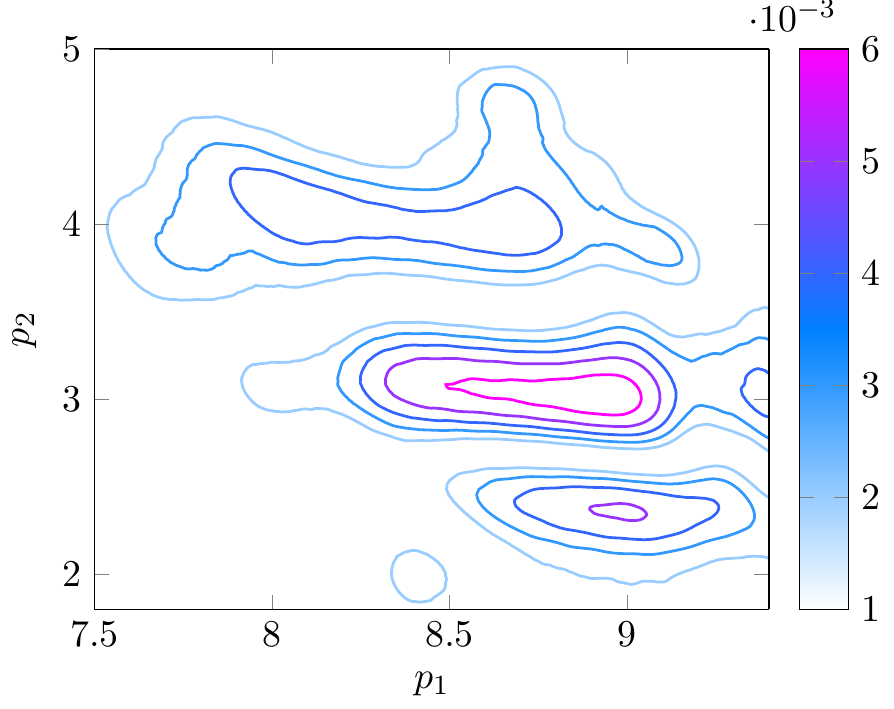}\label{fig:CJLMFormulaResult2NG}}\\
			\caption{As in \cref{fig:CJLMFormulaResult}, but with initial wavepacket \cref{eqn:NonGaussian}.}\label{fig:CJLMFormulaResultNG}			
		\end{figure}					
		By using instead the parameters
		\begin{align}\label{Parameters3}
		\left\{\varepsilon,\delta,\bm{p}_0,\bm{x_0}\right\} &= \left\{\frac{1}{30},0,(5,0),(0,0.5)\right\}
		\end{align}	
		we consider the Jahn-Teller potentials, which include a conical intersection. We have chosen momentum such that the centre of mass of the wavepacket does not cross the intersection. We evolve back to $-25/\|\bm{p}_0\|^2$ with a time-step of $1/(50*\|\bm{p}_0\|)$, then evolve forwards to $25/\|\bm{p}_0\|^2$ using the algorithm, and compare with the exact calculation. Then $Er_{\mathrm{rel}} = 0.0650, Er_{\mathrm{abs}}= 0.0563$, and $Er_{\mathrm{mass}}= 0.0309$ for initial wavepacket of form $\psi_0$ and $Er_{\mathrm{rel}} = 0.1532, Er_{\mathrm{abs}}= 0.0884 $, and $Er_{\mathrm{mass}}=0.0606$ for $\phi_0$, the transmitted wavepacket and error is given in \cref{fig:CJLMFormulaResult}. Although the relative error is large in this final calculation, the absolute error and mass error shows that the algorithm has performed well, given that it is not designed for systems where $\delta$ is small or vanishing. \cref{fig:JTFormulaResultNG} also shows that the shape of the wavepacket is still well approximated qualitatively. 
		
		We note that the relative and absolute error in \cref{Example:CJLM} differ, while in \cref{Example:tanh} they are the same. We believe this is due to a change in phase when $\rho$ is not flat in $q_2$, so the error due to the modification of the potential surface for each strip is larger.
		\begin{figure}[ht!]
			\centering
			\subfloat[][]{\includegraphics[width=0.5\textwidth]{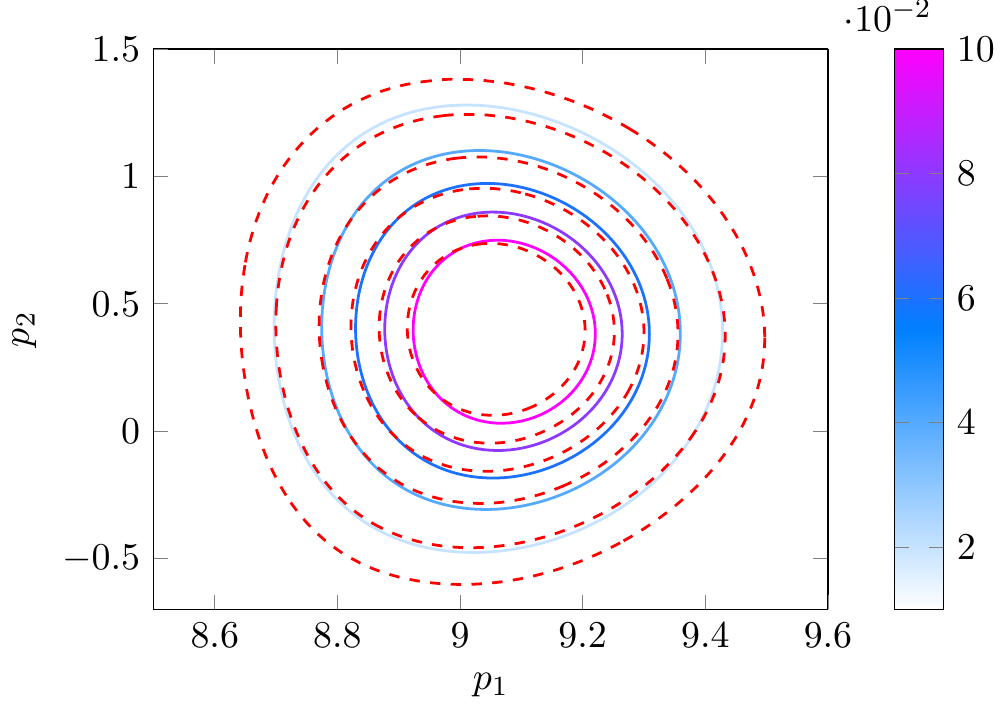}\label{fig:JTFormulaResult1}}
			\subfloat[][]{\includegraphics[width=0.5\textwidth]{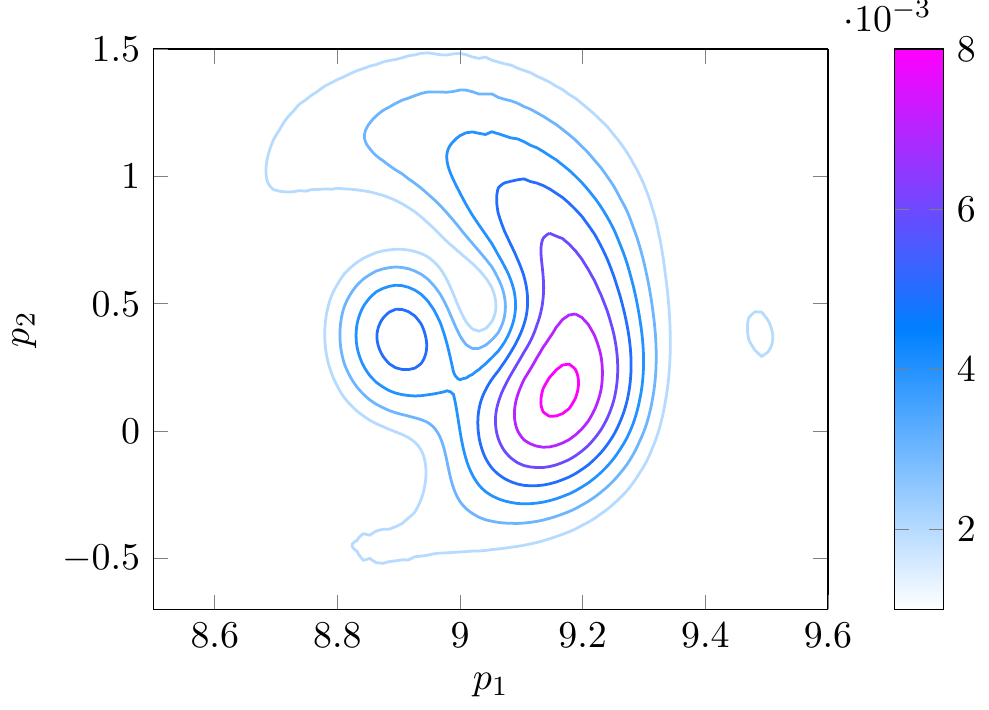}\label{fig:JTFormulaResult2}}\\
			\caption{As in \cref{fig:CJLMFormulaResult}, but with parameters \cref{Parameters3}.}\label{fig:JTFormulaResult}			
		\end{figure}
		\begin{figure}[ht!]
			\centering
			\subfloat[][]{\includegraphics[width=0.5\textwidth]{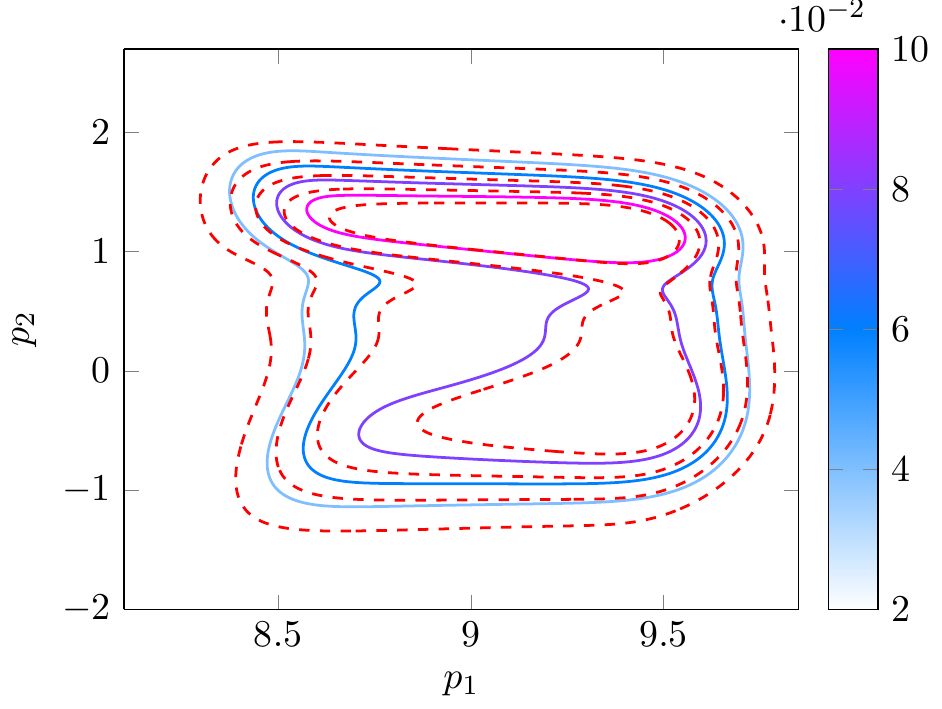}\label{fig:JTFormulaResult1NG}}
			\subfloat[][]{\includegraphics[width=0.5\textwidth]{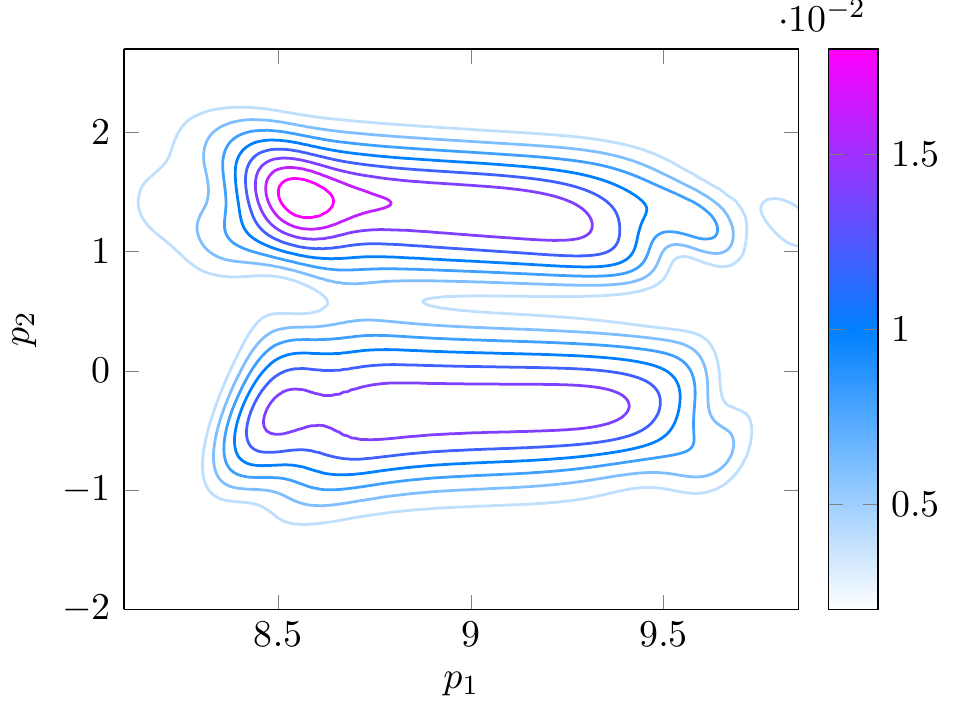}\label{fig:JTFormulaResult2NG}}\\
			\caption{As in \cref{fig:JTFormulaResult}, but with initial wavepacket \cref{eqn:NonGaussian}.}\label{fig:JTFormulaResultNG}			
		\end{figure}
	\end{example}

	\section{Conclusions and Future Work}\label{Section:Conclusions}	
	In this paper we have constructed an algorithm which can be used to approximate the transmitted wavepacket in non-adiabatic transitions in multiple dimensions, by constructing a formula based on the one dimensional result in \cite{BetzGoddard09}, and appealing to the linearity of the Schr\"{o}dinger equation to decompose the dynamics onto strips with potentials that are constant in all but one direction. Presented examples in two dimensions show similar accuracy to one dimensional analogues, and are accurate in the phase, which is beyond the capability of standard surface hopping models.
	
	Correctly approximating the phase of the wavepacket becomes important when more than one transition takes place. In \cite{GoddardHurst18} various one dimensional examples of multiple transitions are explored using the formula, with accurate results. In future work we will consider multiple transitions in two dimensions using the algorithm. This will involve taking into account the effect of geometric phase \cite{GeometricPhaseBook} due to multiple avoided crossings, as well as constructing an approximation of the wavepacket which remains on the upper level after a transition has taken place.
	\bibliographystyle{siamplain}
	\bibliography{2DQMDBibliography.bib}
\end{document}